\documentclass[twocolumn, twoside, a4paper, 9pt]{article}

\usepackage[utf8x]{inputenc}
\usepackage[switch,columnwise]{lineno}

\usepackage[a4paper,
            left=0.75in,
            right=0.75in,
            top=0.75in,
            bottom=1in,
            footskip=.25in]{geometry}

\usepackage{fancyhdr}
\usepackage{enumitem}
\usepackage{bm}
\usepackage{amsmath, amssymb, amsthm}
\usepackage{graphicx}
\usepackage{eurosym}
\usepackage{authblk}

\newcommand{\ie}{i.e.,}
\newcommand{\eg}{e.g.}

\newcommand{\var}[1]{\textsf{\hbox{#1}}}
\newcommand{\vbm}[1]{\textsf{#1}}

\usepackage[ruled,vlined,linesnumbered,noend]{algorithm2e}

\newlength\mylen

\DeclareMathOperator*{\argmax}{arg\,max}

\newtheorem{ourClaim}{Proposition}
\newtheorem{definition}{Definition}
\newtheorem{remark}{Remark}
\newtheorem{corollary}{Corollary}

\newcommand{\avgN}{\bar{\delta_\Delta}}

\usepackage{tikz}
\usepackage{threeparttable}
\usepackage{multirow}

\usepackage{tcolorbox}

\newcommand{\updated}[1]{\textcolor{black}{#1}}

\definecolor{comment-color}{HTML}{606060} 

\title{Geographical Peer Matching for P2P Energy Sharing}

\author[1]{Romaric Duvignau}
\author[1]{Vincenzo Gulisano}
\author[1]{Marina Papatriantafilou}
\author[2]{Ralf Klasing}
\affil[1]{Chalmers University of Technology (CTH), Sweden}
\affil[ ]{\texttt \{duvignau,vinmas,ptrianta\}@chalmers.se}

\affil[2]{CNRS, LaBRI, Universit\'e de Bordeaux, Talence, France}
\affil[ ]{\texttt ralf.klasing@labri.fr}

\begin{document}


\maketitle

\begin{abstract}
Significant cost reductions attract ever more households to invest in small-scale renewable electricity generation and storage. Such distributed resources are not used in the most effective way when only used individually, as sharing them provides even greater cost savings. Energy Peer-to-Peer (P2P) systems have thus been shown to be beneficial for prosumers and consumers through reductions in energy cost while also being attractive to grid or service providers. However, many practical challenges have to be overcome before all players could gain in having efficient and automated local energy communities; such challenges include the inherent complexity of matching together geographically distributed peers and the significant computations required to calculate the local matching preferences. Hence dedicated algorithms are required to be able to perform a cost-efficient matching of thousands of peers in a computational-efficient fashion. We define and analyze in this work a precise mathematical modelling of the geographical peer matching problem and several heuristics solving it. Our experimental study, based on real-world energy data, demonstrates that our solutions are efficient both in terms of cost savings achieved by the peers and in terms of communication and computing requirements. Our scalable algorithms thus provide one core building block for practical and data-efficient peer-to-peer energy sharing communities within large-scale optimization systems.
\end{abstract}



\section{Introduction}

Renewable electricity generation is becoming more affordable to end-users as the initial investment cost has been drastically cut, thus transforming the traditional residential households from consumers into \textit{prosumers}~\cite{schleicher2012renewables} (capable of producing locally their own electricity).
Sharing  resources, for example solar photovoltaic (PV) panels and battery systems, in \emph{Peer-to-Peer} (P2P) setups can be leveraged as a way to optimize the cost-benefits from those \updated{distributed renewable resources.}
Hence, primarily driven by the interest to reduce even more their  cost, there is interest in P2P energy communities~\cite{hahnel2019becoming}, a high-potential concept that caught up the attention of the research community in the recent years (cf. \updated{\cite{tushar2021peer} and references therein}).
The principle of P2P energy sharing is for different end-users to share their resources locally in groups, in order to reduce their energy bill \updated{while increasing energy efficiency and self-reliance within the communities.}
Locally in each community, energy is either ``traded'' at regular intervals (\eg{} every hour) or ``exchanged for free'' with a later gratification scheme.
P2P energy sharing bypasses the centralized grid by letting the peers cooperate in a distributed fashion in order to share in the best way their energy resources. 
Hence, to lower their cost, households are encouraged to use as much as possible the  electricity that is generated locally instead of that from the grid, with local benefits due to reduced tax fees, as well as an increase in local self-consumption.
However, organizing end-users into communities reveals to be a challenging task.
Short-term communities (lasting for e.g. 10min) are formed by using usually game-theoretic approaches~\cite{tushar2019grid} and only taking into account the current state of the system (\eg{} amount of electricity being produced and consumed by the different end-users, as well as market prices).
\updated{Forming long-term communities (used over months or years \cite{chau2019peer, chau2023approximately}) is an ever more rewarding and challenging task.}
Chau~\textit{et al} \cite{chau2019peer} have investigated \textit{stable} partitioning, \ie{} a given partitioning can be rejected if any group of peers would gain more by forming a different one, similar to the \textit{stable marriage}~\cite{manne2016stable} problem for pairs.
\updated{In \cite{zhou2020multi, chau2023approximately}, the partitioning is expanded to communities of size $k > 2$ and partition-forming algorithms are evaluated for groups of size $2$ or $3$ over a set of $30$ households.
Duvignau~\textit{et al} \cite{duvignau2020small, duvignau2021benefits} show that small-scale communities made of a few peers only (2 to 5) are both efficient in terms of data and cost. 
As the authors  essentially focused on energy cost-optimization (increasing benefits for end-users) and data-efficiency (decreasing amount of shared data), no mathematical analysis is performed concerning the cost of computing the different matchings.}
In all the aforementioned works, the computed long-term partitions in the experimental study do not involve a high number of nodes (100 at maximum) or larger neighborhoods than those of size $3$, nor do they use information about the geographical positions of the nodes.

\paragraph*{Motivations}

Current state-of-the-art~\cite{heinisch2019organizing,chau2023approximately,duvignau2021benefits} relies on exhaustive search to compute the optimal solution from datasets containing only a small number of nodes.
To scale up towards large systems, we introduce the \textit{Geographical Peer Matching} (GPM) problem that consists in forming the energy communities based on both geographical information about the peers as well as their local matching preferences.
This provides a natural way of reducing the search space but as we show in this work, dedicated algorithms are \updated{still} required to cope with the computational complexity of the GPM problem.

\paragraph*{Challenges and Research Questions}

The main challenges in establishing efficiently  P2P energy sharing communities are threefold: (i) peers continuously produce data and have limited knowledge of how their future local data will look like, 
(ii) computing peers' optimization options for preferences requires both access to the relevant data and the execution of a 
\updated{LP solver on a large input (cf.~\cite{duvignau2020small,duvignau2021benefits} and Section~\ref{subsec:LPsolver})}
and
(iii) peers should favor getting matched with geographically neighboring peers to reduce transmission losses and the impact on the underlying infrastructure.
In particular, weights in the matching (\ie{} the local matching preferences) are not 
\updated{known at the system's start}
but must be computed on the fly, and this requires additional communication between the  participants.
An additional difficulty stems from allowing the formation of \updated{groups of size $4$ and above}, as the matching problem becomes intractable in this case and thus in practical systems, it is not feasible to do exhaustive searches any longer.
%
%
This raises the following research questions: (i) Can one translate the challenges to a formal model that can capture the benefits and complexity for the prosumers? (ii) How do the maximum size and geographical diameter for communities influence the cost-efficiency of the peer matching? (iii) Is it possible to design matching algorithms that are efficient in terms of \updated{cost (with good quality for the solution) and scalable (with low computational burden) and where lies the best trade-off}?

\paragraph*{Contributions}

\updated{We present partitioning mechanisms between consumers and prosumers to form P2P energy sharing communities.}
One core contribution is the presentation of a mathematical modelling of the GPM problem expressed as finding a maximum weight bipartite partitioning in hypergraphs. 
\updated{With Proposition~\ref{prop:ratio_bounds}, we provide an important theoretical result for the smart grid research community, formally lower- and upper-bounding the benefits that can be extracted from a prosumer by a coordinated group rather than independent individuals.}
We further propose and analyze different algorithms that solve the GPM problem. 
We study the cost-efficiency of the different algorithms based on an experimental study involving consumption data from 2221 real households and realistic solar profiles (both over a year), and using a realistic distribution of renewable energy resources among the peers.
Our findings highlight that our solution for P2P energy sharing communities is both cost-efficient (providing significant cost savings to all peers) and scalable (capable of scaling to pools of at least thousands of users).

\paragraph*{Plan}

Section~\ref{sec:system_model} presents in more detail P2P energy sharing and cost-optimization of distributed resources in this context. Section~\ref{sec:algorithmic_formulation} presents our mathematical modelling of the GPM problem and some of its variants. 
In Section~\ref{sec:algorithms} we present our algorithms and analyze their computational overhead.
In Section~\ref{sec:evaluation}, we present an extensive experimental study of the performance of the algorithms based on electricity data from real-world households. 
The  section following that discusses related work, whereas Section~\ref{sec:conclusion} presents our conclusions.

\section{System Model} \label{sec:system_model}

We define here the notion of P2P energy sharing communities, its underlying assumptions and consequences. 
We then define how to optimize peers cost in such a context and present how this translates into a matching problem when several communities are considered.

\begin{figure}[t]
    \centering
    \includegraphics[width=0.95\linewidth]{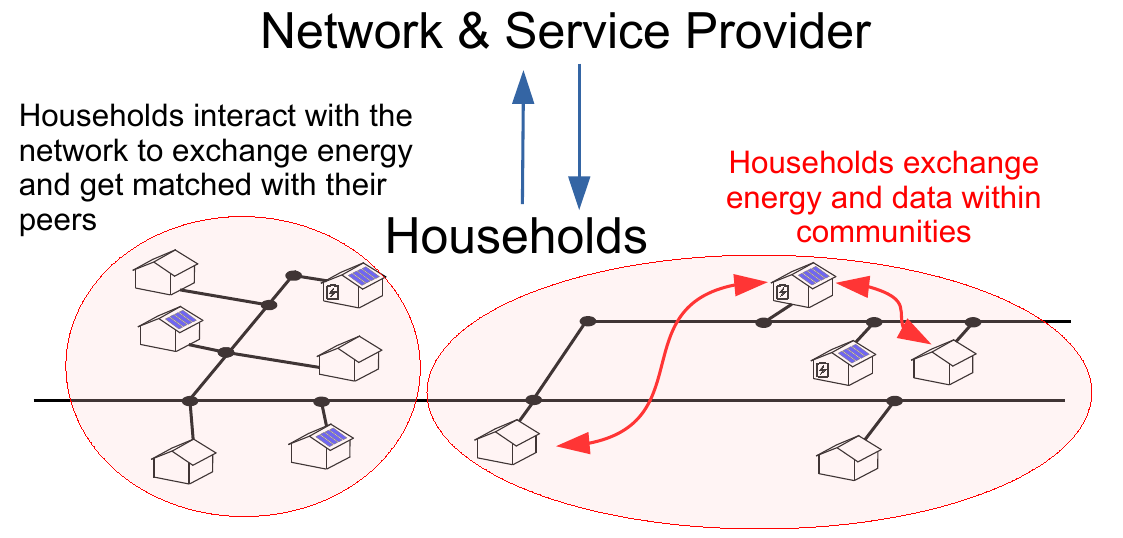}
    \caption{Overview of P2P energy sharing showing grouping and interactions between prosumers (equipped with PV panels on roof top and optionally a battery system) and traditional consumers (without any energy resources).}
    \label{fig:simple_fig}
\end{figure}

\subsection{P2P Energy Sharing}

\begin{table*}[t]
    \updated{
    \caption{Nomenclature used in the paper (P2P Energy Sharing variables, constants and functions in the left column, graph symbols/functions in the right column).} 
    \label{tab:nomenclature}
    \centering
    \fbox{
    \resizebox{\textwidth}{!}{%
    \begin{tabular}{l l | l l}
        \textbf{Symbol} & \textbf{Usage} & \textbf{Symbol} & \textbf{Usage} \\
        \rule{0pt}{3ex}$\var{cost}(h,t)$ & electricity cost for user $h$ at hour $t$ (€) & $V$ & pool (set of all households), $V = P \cup C$  \\ 
        $\var{bill}(h, [t_0, t_r])$ & electricity cost for $h$ over period $[t_0, t_r]$ (€) & $P$ & set of prosumer households \\ 
        $\vbm{el}_{in}(h,t)$ & amount of electricity bought from the grid by $h$ at hour $t$ (kWh) & $n$ & number of prosumers, \ie{} $|P|$ \\
        $\vbm{el}_{out}(h,t)$ & amount of electricity sold to the grid by $h$ at hour $t$ (kWh) &  $C$ & set of consumer households \\
        $\var{el}_{cons}(h,t)$ & consumption (or electricity demand) of user $h$ for hour $t$ (kWh) & $m$ & number of edges, \ie{} $|P| \cdot \avgN$ \\
        $\var{el}_{gen}(h,t)$ & electricity generated by $h$ during hour $t$ (kWh) & $\Delta$ & geographical search radius \\
        $\vbm{bat}(h,t)$ & battery level at time $t$ for prosumer $h$ (kWh) & $k$ & maximum size for communities \\
        $\var{price}(t)$ & price of electricity at hour $t$ (€/kWh) & $\avgN$ & average neighborhood size \\
        $\vbm{sun}(t)$ & the sun's intensity at hour $t$ (kWh/kWp) & $w(e)$ & weight of the hyperedge $e$ \\
        $\var{tax}$ & relative tax level on top of market-price (e.g. 25\%) & $\textsf{dist}(v,v')$ & geographical distance between $v$ and $v'$\\
        $\var{el}_{tax}$ & fixed electricity tax added on top of market price (€/kWh) &$E_\Delta$ & pairs of $E \subseteq P \times C$ within distance $\Delta$ \\
        $\var{el}_{net}$ & small payment for selling electricity to the grid (€/kWh)& $\var{WA}_t$ & cost-based memoryless weights at time $t$\\
        $\var{PV}_h$ &  PV capacity for prosumer $h$ (kWp) & $\var{WB}_t$ & cost-saving memoryless weights at time $t$\\
        $\var{B}_h$ & battery capacity for prosumer $h$ (kWh) & $\var{WC}_t$ & cost-based memoryful weights at time $t$\\
        $\var{gain}(G, [t_0, t_r])$ & cost saving of community $G$ over period $[t_0, t_r]$ (€) &$\var{WD}_t$ & cost-saving memoryful weights at time $t$\\
    \end{tabular}
    }
    }
    }
\end{table*}

Let us recall the basic requirements of the traditional energy infrastructure, 
\updated{considering only electricity as energy}
for the purpose of simplicity.
Consumers must match their consumption by importing from the grid their electricity demand, whereas prosumers 
use in priority their local generation (\eg{} PV panels); 
in case of surplus, the energy is sold to the grid, while in the opposite situation, prosumers must import electricity from the grid.
The situation complicates for prosumers equipped with both electricity production and storage, 
as they need to take online decisions whether to store the surplus or sell it to the grid, and whether to use stored electricity or rather buy it from the grid.

Set in the context of increasing decentralization of the energy infrastructure, \textit{P2P energy sharing} (see Fig.~\ref{fig:simple_fig}) consists in forming local \textit{energy communities} of cooperative end-users.
The goal is to make the most of the distributed resources and hence achieve even greater cost reductions.
However, this means that participants need now to consider also the state and decisions of the other actors in order to coordinate and optimize the benefits of their local resources.
\updated{In such energy communities, any energy consumption can be offset by importing the equivalent amount of electricity from a peer; such an exchange may get instantly gratified leading to a local trading market~\cite{paudel2018peer,zhang2019two}. At regular time intervals (\eg{} 1 hour), the community needs to coordinate the usage of energy resources (e.g. which battery system(s) to charge or discharge and by how much, and by consequence how much energy needs to be traded with the central grid).}
Since there is an inherent infrastructure cost to allow and maintain energy exchanges \updated{among} a large number of end-users, exchanges have been restrained to occur within independent \textit{long-term communities} of end-users~\cite{long2018peer, heinisch2019organizing, chau2023approximately, duvignau2021benefits} that can last for month(s) or year(s).
At the end of predetermined billing periods, each user pays an electricity bill \updated{taking into account their own consumption and all exchanges that occurred within the said period.}

\subsection{Single-user Cost-Optimization (LP-solver)} \label{subsec:LPsolver}

The cost-optimization problem consists in minimizing the yearly electricity bill for a particular end-user, based on locally available electricity data: amount of consumption, generation and price. 
In order to minimize the cost, we adopt the following LP-formulation cost-optimization following similar models used in the recent literature (cf.~\cite{chau2019peer,heinisch2019organizing,duvignau2021benefits}).
The electricity cost $\var{cost}(h,t)$ of the end-user $h$ at hour $t$ is assumed to be as follows \updated{(cf. Table~\ref{tab:nomenclature} for definitions of all variables used hereafter)}:
\vspace{-0.25cm}
\begin{multline}\label{eq:cost} \medskip
   \var{cost}(h,t) = \vbm{el}_{in}(h,t) \cdot \left( \var{price}(t) \cdot (1+\var{tax}) + \var{el}_{tax} \right) \\- \vbm{el}_{out}(h,t) \cdot \left( \var{price}(t) + \var{el}_{net} \right), 
\end{multline}
%
We assume $\var{el}_{net} < \var{el}_{tax}$ and $\var{tax} \geq 0$.
Consumers do not have any resources, hence the yearly cost is obtained directly from their consumption, that is $\vbm{el}_{in}(h,t) = \var{el}_{\hbox{cons}}(h,t)$ for all hours $t$. Prosumers with only electricity generation but no storage have always interest to use in priority their local production to avoid to pay tax on electricity coming from the grid; hence, they optimize their cost by setting:
\hspace{-0.25cm}
$$
    \vbm{el}_{in}(h,t) = 
\begin{cases}
    \var{el}_{gen}(h,t) - \var{el}_{cons}(h,t),& \text{if } x_{h,t}>0,\\
    0,              & \text{otherwise;}
\end{cases}
$$
$$
    \vbm{el}_{out}(h,t) = 
\begin{cases}
    \var{el}_{cons}(h,t)-\var{el}_{gen}(h,t),& \text{if } x_{h,t}<0,\\
    0,              & \text{otherwise;}
\end{cases}
$$


\noindent with the electricity balance $x_{h,t} = \var{el}_{gen}(h,t) - \var{el}_{cons}(h,t)$ where $\var{el}_{gen}(h,t) = \var{PV}_h \cdot \vbm{sun}(t)$. 
Prosumers having both electricity generation and storage can optimize their cost over a period of time from $t_0$ to $t_r$ through running an LP solver of the following formulation:
\begin{itemize}
    \item \textbf{Objective function:} \\
    minimize $\var{bill}(h, [t_0, t_r]) = \sum_{t = t_0}^{t_r} \var{cost}(h,t)$. 
    \item \textbf{Constraints (for all $t_0 \leq t \leq t_r$):} \\
    $0 \leq \vbm{bat}(h,t) \leq \var{B}_h$,  and\\
    $\vbm{bat}(h,t) = \vbm{bat}(h,t-1) +  \var{el}_{gen}(h,t) - \var{el}_{cons}(h,t) + \vbm{el}_{in}(h,t) - \vbm{el}_{out}(h,t).$ 
    \item \textbf{Optimization variables:} \\
    $\{\vbm{bat}(h,t), \vbm{el}_{in}(h,t), \vbm{el}_{out}(h,t) \;|\; t_0 \leq t \leq t_r\}$. \smallskip
\end{itemize} 
\vspace{-0.1cm}
with $\vbm{bat}(h,t_0-1)$ indicating the initial battery level. 

\updated{This simplified model introduced in~\cite{duvignau2021benefits} as ``aggregate model'' does not account for transmission and battery losses. 
Note that in order to take continuous online decisions concerning usage of one's battery system, one would need to use forecast data as input to the optimization problem as in~\cite{long2018peer,duvignau2020small,duvignau2021benefits}.
}

\subsection{Community Cost-Optimization} \label{subsec:community}

One may wonder how the end-users can in our context make the most of their distributed resources. 
Instead of fixing one of the many forms of a local trade market, we consider as in e.g.~\cite{heinisch2019organizing,duvignau2020small,duvignau2021benefits} that the energy exchanges occur ``for free'' within the community to obtain and analyze the lowest achievable cost as a community, while postponing the billing of individual exchanges to the end of the billing period. 
\updated{Hence, our focus in this work is on finding communities that reach the best benefits overall. A subsidiary mechanism can occur  when the billing period ends to distribute the gain achieved by the community among the peers. }
Under the above assumptions and neglecting battery and transmission losses and communication issues, each community becomes equivalent to a single prosumer with aggregated PV and battery capacities over the full community. 

\paragraph*{\updated{Example}}

Let us consider as an example a particular grouping of 6 households $V = P \cup C$, including 1 battery-equipped prosumer $p_1$ among 3 prosumers $P = \{p_1,p_2,p_3\}$ and 3 consumers $C = \{c_1,c_2,c_3\}$, as the one on the left of Fig.~\ref{fig:simple_fig}.
Neglecting losses and communication faults, the community is then equivalent to a single larger entity with aggregated consumption from the 6 households, aggregated production from the 3 prosumers and having as much \updated{battery} storage as $p_1$.  
\updated{
Suppose now the annual cost (i.e. $T = [t_0,t_r]$ spans one year) of each household $h \in V$ was $\var{bill}(h, T)$ = 1000€ each, for a total of 6000€. 
As a single community taking coordinated decisions, they may only have to pay 4800€, or 800€ each if the 1200€ gain is spread equally among participants.
We use hereafter \textit{cost saving} for the reduction in cost obtained through cooperation, \eg{} 1200€ in this example.
}

\updated{
\vspace{-0.3cm}
\begin{definition}\label{def:cost_savings}
    Let the {\em cost saving} (or {\em gain}) for the community $G$ be defined as $$\var{gain}(G, [t_0, t_r]) = \sum_{x \in G} \var{bill}(x, [t_0, t_r]) - \var{bill}(\bar{G}, [t_0, t_r]),$$ where $\bar{G}$ is an aggregated prosumer equivalent to $G$, defined by: $\var{PV}_{\bar{G}} = \sum_{x \in G} \var{PV}_x$, $\var{B}_{\bar{G}} = \sum_{x \in G} \var{B}_x$, and for each $t \in [t_0, t_r]$, we have $\var{el}_{gen}(\bar{G},t) = \sum_{x \in G} \var{el}_{gen}(x,t)$,  $\var{el}_{cons}(\bar{G},t) = \sum_{x \in G} \var{el}_{cons}(x,t)$.
\end{definition}
}

\subsection{From Communities to Partitions}

Among a \textit{pool} (set of end-users willing to participate in P2P energy sharing), several communities can be managed independently from each other. 
How to partition efficiently a group of users into independent communities is the main focus of the present work.
For any given partitioning \updated{(e.g. the partition of the households of Fig.~\ref{fig:simple_fig} into $2$ communities, one of size $6$ and one of size $5$)}, we can associate a \textit{global cost saving} corresponding to the sum of the cost savings of each community.
Optimizing cost-efficiency of all resources in a given pool corresponds, from a \updated{centralized} perspective, to maximizing the global cost saving.
\updated{One can thus summarize the problem of forming P2P energy communities as follows:
\begin{itemize}[leftmargin=*]
    \item[-]\textbf{Input:} (forecast or historical data for) (1) energy consumptions over  timespan $\mathcal{T} = [t_0,t_r]$ for each household $h \in V$, where $V = P \cup C$ is made of a set $P$ of prosumers and $C$ of consumers; (2) local solar intensity (depending on geographical location) over $\mathcal{T}$; (3) electricity prices (possibly set at regional level) over $\mathcal{T}$.
    \item[-]\textbf{Ouput:} A partition $M$ of $P \cup C$ into independent groups (whose size may be limited by a certain constant $k$).
    \item[-]\textbf{Metrics of interests:} (1) amount of global cost saving obtained by the partition $M$ (that is, the sum of the communities' gains, i.e., $\sum_{G \in M} \var{gain}(G, [t_0,t_r])$, cf. Definition~\ref{def:cost_savings}); (2) computational overhead of calculating $M$. 
\end{itemize}
}

\subsection{From Partitions to Geographical Peer Matchings} \label{subsec:peer_matching}

\updated{
The search space for possible partitions of the pool of households into independent communities is large and much beyond computational capabilities of any real system of reasonable size. 
We hence propose several basic restrictions on the possible partitions to reduce the combinatorial possibilities.
First, communities made uniquely of consumers and single-node communities can be discarded as they provide absolutely no benefit in terms of cost-saving.
Second, following results from previous works in the area~\cite{heinisch2019organizing,duvignau2021benefits}, we restrict communities to be made of a single prosumer as allowing several prosumers provides very little to no advantage in terms of global cost saving in comparison to splitting such group into two or more smaller communities.
Third, we propose to limit communities to be only made of nearby households by restricting the maximal distance between the prosumer and any consumer in a given community to be within a certain constant $\Delta$ (hence any household of a community is always within $2\Delta$ of each other).
Limiting the search radius to $\Delta$ allows the service provider managing the P2P energy system to use aggregators\footnote{intermediate infrastructure level between end-users and service provider where data can be retrieved almost in real time and with fine granularity.} in charge of smaller geographical areas. Remote control of the end-users' distributed resources could in turn be delegated to such aggregators relieving the users of data exchange and computational work during the P2P energy sharing process.
In addition, having geographically closer communities allows to have better independence and load-balance in the system, further reducing the impact of introducing sharing communities on the underlying grid infrastructure.
Hence, the problem of forming communities reduces to finding a \textit{matching} (or assignment) of nearby consumers to each prosumer in the system, and we thus refer to it as the \textit{Geographical Peer Matching} (GPM) problem, further formalized and defined in Section~\ref{sec:algorithmic_formulation}.
}

\updated{ 
%
%
In the GPM problem, peers are matched together based on their current \textit{preferences} $w_t(G)$ at time $t$, which indicate e.g. the potential saving of a certain community $G$ at time $t$. 
Since data is not known ahead of time, there are two strategies to compute $w_t(G)$: either using only past data, \ie{} $w_t$ is computed based on data recorded within some timespan $[t-\tau,t]$ for some $\tau$, or using past and projected data. 
Since we consider in this work communities lasting for long period of time (months to years), the preference calculation based on projection is not the most appropriate solution as the accuracy of the forecast degrades fast as the horizon grows (\eg{} poor prediction is expected past 48h).
In this context, peers affinity is more reasonably captured by setting $\tau$ to the same length as the billing period (e.g. one year).
\updated{We would like to highlight that a long time horizon for the duration of the energy communities does not imply that one can disregard the computational burden of calculating the peer matching. Indeed, the GPM problem is computationally hard to solve (cf.\S~\ref{subsec:complexity_measure}) and calculating an optimal solution for a large instance is considerably beyond the capacity of the computing resources available at the service provider for P2P energy-sharing.}
In addition, in order to compute $w_t(G)$, necessary data must be transmitted (enduring some communication overhead) and a run of a LP-solver (requiring computation overhead) is required to solve the cost-optimization problem (see~\S~\ref{subsec:LPsolver} and \S~\ref{subsec:community}). 
Hence, the computation of a single preference $w_t(G)$ is therefore both costly in terms of data exchanged and local computation. 
The goal of our approach is thus for the group matching to be computed while minimizing as much as possible computation of such preferences. 
}

\subsection{General System Considerations} 

We assume customers can see and use their own data, but do not have access to their peers' data.
To minimize data exchanges and reduce stress on the architecture, we assume 
the matching happens at a higher level through a third party dedicated entity.
This centralized point of view assumes the end-users have  subscribed to such an external service provider (being the energy provider or a third agent) in order to participate in the sharing process, and paying for the service through a \updated{fixed} share of the cost saving obtained by each community. 
Under this setting, the service provider is in charge of grouping prosumers and consumers and supervising the transmission of data needed for the matching; its goal is thus to achieve the best global cost saving to maximize its own benefit as well.
\updated{This also means that, once a matching has been decided and propagated to the peers by the system, communities can then work in an independent fashion. In particular, they do not further need to rely on the service provider for managing their every day exchanges, neither for optimizing their electricity cost.}
Since end-users change their consumption patterns and may revoke their will to participate through time,
it may be beneficial to recompute the matching after a billing period has elapsed.
\updated{
In our experimental evaluation presented in Section~\ref{sec:evaluation}, we have used historical data for one year as input to the aggregate model (LP-optimization) to calculate the peers' matching preferences. 
This is based on the assumption that the peers participating in long-term energy communities will likely reproduce overall the same or similar consumption patterns over time (and entailing similar matching preferences in the future).
}

\begin{figure*}[t]
    \centering
    \includegraphics[width=\linewidth]{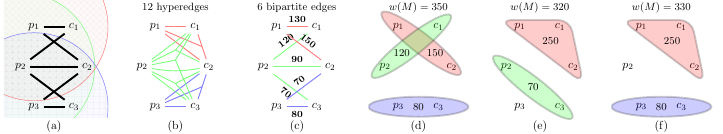}
    \caption{Illustration of the matching procedures with $P = \{p_1, p_2, p_3\}$, $C = \{c_1, c_2, c_3\}$ and $k = 3$: (a) authorized edge set $E_\Delta$ for a given search radius $\Delta$, (b) all hyperedges based on $E_\Delta$, (c) pairwise weights, and (d) Round Robin, (e) Single Pass, and (f) Classic Greedy matchings.}
    \label{fig:matchings}
\end{figure*}


\section{Algorithmic Modeling for Geographical Peer Matchings} \label{sec:algorithmic_formulation}

We present in this section a formalism for the GPM problem in terms of finding a maximum-weight matching in bipartite hypergraphs. \updated{We then present how to measure the complexity of algorithms solving the problem and how certain considerations on the behaviors of the peers' preferences can simplify the GPM problem.}



\subsection{The GPM Problem Abstraction} \label{subsec:pb_abstraction}
\subsubsection{Preliminaries}

A hypergraph $\mathcal{G} = (V,E)$ is made of a set of vertices $V$ and a set of hyperedges $E \subseteq 2^V \setminus \emptyset$, where a hyperedge $e \in E$ is any non-empty subset of vertices.
A hypergraph is said to be \textit{weighted} if the hypergraph is associated with a weight function $w : E \rightarrow \mathbb{R}$.
We say a hypergraph is $k$-bounded if all its hyperedges are of size at most $k$; a $3$-bounded hypergraph is displayed in Fig.~\ref{fig:matchings}(b).

A matching $M$ in a hypergraph $\mathcal{G} = (V,E)$ is a set of \textit{disjoint hyperedges} ($e_1$ and $e_2$ are considered disjoint when they do not share any vertices, \ie{} $e_1 \cap e_2 = \emptyset$). 
The weight of such a matching is the sum of the weights of the selected hyperedges that it contains, \ie{} $w(M) = \sum_{e \in M} w(e)$, slightly abusing the $w$-notation. 
Given a partition of the vertices into two disjoint sets $P$ and $C$, a \textit{bipartite hypergraph matching} (BHM) is a matching of a hypergraph $\mathcal{G} = (P \cup C, E)$ that contains in each selected hyperedge exactly one vertex within $P$, \ie{} $M$ is bipartite if $\forall e \in E$, $|e \cap P| = 1$.
A $k$-bounded matching ensures that each output group is of size at most $k$; $3$-bounded BHMs are displayed in Fig.~\ref{fig:matchings}(d)-(f). \looseness=-1


\subsubsection{Geographical Peer Matching (GPM)} 

We define the \textit{Maximum-Weight Bipartite Hypergraph Matching} (MWBHM) problem that consists in finding a bipartite hypergraph matching with maximum weight in a given hypergraph over vertex set $P \cup C$. 
The problem becomes dynamic when the weighting function $w$ is time-dependent $w_t$ and the matching problem should be solved for each time step.
The \textit{GPM} problem of parameter $(k,\Delta)$\updated{, for $k \geq 2$ and $\Delta > 0$, written in short form as $(k,\Delta)$-GPM,} is a bounded version of the MWBHM problem where one adds the following three additional conditions: %
\begin{enumerate}[leftmargin=*,topsep=0pt]
\item \textbf{Neighborhoods}: $M$ is a $k$-bounded bipartite hypergraph matching, \ie{} $\forall e \in M, |e| \leq k$.

\item \textbf{Spatiality}: peers are geographically distributed and the matching $M$ should adhere to each peer's locality: 
every consumer of a selected hyperedge of $M$ must be within geographical distance $\Delta$ of its matched prosumer, \ie{} $\forall e \in M, \{p\} = e \cap P, \forall c \in e \cap C, \; \textsf{dist}(p,c) \leq \Delta$. 

\item \updated{\textbf{Computationally-intensive weights}: weights are assumed unknown beforehand but must be dynamically calculated (cf. \S~\ref{subsec:peer_matching}), and the function is considered expensive (due to inherent communication and computational costs) and the main bottleneck of the system. Complexity of an algorithm solving the GPM problem is thus mainly measured in the number of $\textsf{weight}$ computations, as further explained in \S~\ref{subsec:complexity_measure}.}
\end{enumerate}


\subsubsection{\updated{P2P Energy Sharing as a Hypergraph Matching Problem}} \label{subsec:p2phypergraph}

We model the problem of forming P2P energy sharing communities in a continuous fashion as a (dynamic) GPM problem of parameters $k$ and $\Delta$. 
The underlying $k$-bounded hypergraph is $\mathcal{G} = (V, \mathcal{E}_\Delta)$ where the vertex set $V = P \cup C$ is made of $P$, the set of prosumers (users equipped with renewable energy resources) and~$C$, the set of consumers (with no resources). 
The set of hyperedges $\mathcal{E}_\Delta$ is made of all possible hyperedges that can be part of a $k$-bounded BHM and consumer-prosumer pairs are pairwise within distance $\Delta$ (where $k$ represents the maximum allowed size for the communities), \ie{} $\mathcal{E}_\Delta \subseteq \{ e = (p,X) \in P \times 2^C \;|\; 2 \leq |e| \leq k, \; \Delta \geq \max_{c \in X} \textsf{dist}(p,c)   \}$.
The weighting function $w_t(G)$ is dynamic and can be used to \eg{} capture the cost saving of a particular community $G$ at time $t$ (cf. \S~\ref{subsec:weights}). 
\updated{Then, adding the hyperedge $G$ in the BHM is equivalent of forming the P2P energy sharing community $G$.} 

\subsection{Complexity of the GPM problem} \label{subsec:complexity_measure}

\subsubsection{Complexity Measures for the GPM problem}

\updated{
The GPM problem belongs to the family of \textit{Discovery Problems}~\cite{duvignau2023greediness} where the input is not entirely accessible at the algorithm's start but must be queried during its execution.
In the case of the GPM problem, the discovery deals with the weight of hyperedges (\ie{} possible communities in the P2P energy sharing context) and as the main bottleneck in computing a solution, the number of calls to the $\textsf{weight}$ function is one of the main complexity measures to evaluate an algorithm, along with the quality of the solution. 
Indeed, one should not forget that GPM is essentially a sub-variant of MWBHM and hence its primary goal is still to find the best possible matching of the hyperedges.
GPM can therefore be considered as a \textit{bicriteria} optimization problem: maximize the weight of the matching and minimize the number of calculated weights.
}
%
%
Hence, considering the difficulty of finding the optimal solution (cf. Proposition~\ref{proposition1}) and the additional assumption on computational-intensive weights, an algorithm $\mathcal{A}$ that produces a matching $M$ as solution to the GPM should be evaluated on two main criteria: (1) the \textbf{number of weights} that were computed by $\mathcal{A}$ in order to find $M$, and (2) the \textbf{quality of the solution} indicated by $w(M)$.
\updated{
One should observe that minimizing the number of weight calculations alone is trivial by just producing an arbitrary valid matching without calling the weight function. 
Let's note that finding the minimum number of weights that need to be computed to reach a fraction $\alpha$ of the weight of the optimal matching is out of the scope of the present work.
}

\updated{
\vspace{-0.3cm}
\begin{definition}\label{def:neighborhoods}
    We call {\em neighborhood} of node $v \in P$ the set $$N_\Delta(v) = \{v' \;|\; \textsf{dist}(v,v') \leq \Delta\},$$ and denote the {\em average neighborhood size} by $$\avgN = \frac{1}{|P|} \sum_{v \in P} |N_\Delta(v)|,$$ used for complexity computations in the remaining.
\end{definition}
}

\paragraph*{Example}
For instance, assuming that the nodes are spread uniformly and independently on a square zone of size $L \times L$ with toroidal properties to simplify\footnote{A torus topology means that each node has on average the same number of neighbors regardless of its position on the map and is very similar to the usual map when $L \gg \Delta$.}, \updated{it is easy to calculate the average neighborhood size as $\avgN = (n-1) \cdot \pi (\Delta/L)^2 = \Omega(n)$ for $L = \Omega(\Delta)$ where $n = |P|$ (recall Table~\ref{tab:nomenclature} contains definitions for all variables).}

\subsubsection{Computational complexity for the GPM problem}

\begin{ourClaim} \label{proposition1} 
For $k \geq 4$,  the $(k,\Delta)$-GPM problem is not approximable within a factor of $o(k/\log k)$ in polynomial time, unless P = NP.
\end{ourClaim}

\begin{proof} 
In general, the hypergraph matching (HM) problem that consists in computing the maximum-weight matching of hyperedges does not permit a polynomial time $o(k/\log k)$-approximation unless P = NP~\cite{HazanSS06}.
There is a trivial reduction to the weighted $k$-set packing problem, known to be NP-complete from Garey and Johnson~\cite{garey1979computers}.
Indeed, a hyperedge is nothing more than a subset of the vertex set, and the $k$-set packing problem~\cite{chandra2001greedy} is looking for a maximum weight sub-collection of disjoint sets, which is equivalent to finding a maximum-weight matching of non-overlapping hyperegdes.

Now, adding the bipartite constraint is not reducing the difficulty of the problem, as one can reduce the GPM problem to HM as follows.
Let $\mathcal{G} = (V,E)$ be a $k$-bounded hypergraph and let $\mathcal{G}' = (V_1 \cup V_2, E')$ be a hypergraph that contains all vertices of $\mathcal{G}$ plus $|E|$ ``extra vertices'', \ie{} $V_1 = \{ v_e \;|\; \forall e \in E\}$ and $V_2 = V$. 
Then all hyperedges of $\mathcal{G}'$ are made of those of $\mathcal{G}$ with one extra vertex in each, that is $E' = \{ e \cup \{v_e\} \; | \; e \in E\}$ with $w(e \cup \{v_e\}) = w(e)$ for every $e \in E$. 
Any matching $M$ of $E'$ is only made of edges that contain exactly one vertex of $V_1$ each, hence is a bipartite hypergraph matching under our definition. 
In turn, $M$ is also a matching of $\mathcal{G}$ (just discard the extra vertex in each hyperedge).
Last, we can easily affect a position to each vertex so that all vertices are within distance $\Delta$ from each other for any $\Delta > 0$, hence defeating the additional spatial constraint.
\end{proof}

\updated{We show in Proposition~\ref{proposition1} that for $k \geq 4$, the $(k,\Delta)$-GPM is a computationally intractable problem~\cite{hopcroft2001introduction}, therefore justifying the introduction of dedicated heuristic algorithms in Section~\ref{sec:algorithms}.}
When forgetting the spatial constraint, we note that for $k = 2$, the problem becomes polynomial and is equivalent to finding the maximum-weight matching in a weighted bipartite graph, usually then named ``the assignment problem''. This is a classic problem where the Hungarian algorithm~\cite{kuhn1955hungarian} provides the optimal solution in time $\mathcal{O}(nm + n^2 \log n)$ for $n$ vertices in the smaller vertex set and $m$ edges, cf.~\cite{ramshaw2012minimum}. 
Now adding spaciality and assuming $\avgN$ as the average neighborhood size for search range $\Delta$, this gives a running time of $\mathcal{O}(n^2 \avgN + n^2 \log n) = \mathcal{O}(n^3)$. We note that such a running time can be already prohibitive for large $n$ and even for $k = 2$, the problem becomes more challenging when considering its distributed equivalent~\cite{lotker2015improved}. 
\updated{Efficient discovery algorithms for the assignment problem were recently explored in~\cite{duvignau2023greediness}.
For $k = 3$, the complexity of the GPM problem remains an open issue.}


\subsection{Reduction and Approximation for the GPM Problem} 

\subsubsection{Reduction to a One-to-Many Assignment Problem} \label{subsec:assignment_pb}

Let's assume here that the weight of any group $G = p \cup \{c_1, \dots, c_{\ell}\}$ with $\ell \leq k-1$, $p \in P$ and $c_i \in C$, can be calculated as the sum of the individual \textit{pairwise weights}, \ie{} $w(G) = \sum_{1 \leq i \leq \ell} w(\{p, c_i\})$. 
In this situation, finding the best group (in terms of weight) of size $k$ containing $p$ is then equivalent of picking the $k-1$ best partners for $p$.
Hence, the maximum weight matching can be reduced to a one-to-many assignment problem, that matches members of the set $P$ with at most $k-1$ members of the set $C$ such that the sum of the individual pairwise weights, \ie{} the ``edge weights'' $w(\{p, c\})$, is maximum. 
This problem can be further reduced to the classical and well known one-to-one assignment problem in the following manner: for each $p \in P$, make $k-1$ copies $p^1, \dots, p^{k-1}$ of node $p$, while keeping the original weights, \ie{} $\forall c \in C, w(\{p^j, c\}) = w(\{p, c\})$.
Finally solve the one-to-one assignment problem (maximum matching in bipartite graphs) with the input $P' = \{ p^j \;|\; j \in [1..k-1], p \in P\}$ and~$C$. 
As mentioned in~\S~\ref{subsec:complexity_measure}, the assignment problem can be solved exactly using the Hungarian algorithm in time $\mathcal{O}(|P'|^2 \cdot |C|) = \mathcal{O}(k n^2 \cdot \max\{ \avgN, \log n \}) = \mathcal{O}(k n^3)$, which is already prohibitive for large systems, 
see \eg{}~\cite{cui2016solving}.

We note that setting the original weights (in the hypergraph) to the cost saving does not follow this paradigm and even though finding the maximum matching of the pairs do provide a hypergraph matching, it is not guaranteed any longer to be maximal in terms of the sum of the weights of the hyperedges. 
Such an example is given in Fig.~\ref{fig:matchings}(f) where $M = \{ \{p_1,c_1,c_2\}, \{p_3,c_3\} \}$ maximizes $\sum_{(p,c) \in P \times C \,|\, \exists e \in M, \{p,c\} \in e} w(\{p,c\}) = 360$ but $w(M) = 330$ and thus $M$ is not maximum, cf. the matching of Fig.~\ref{fig:matchings}(d) \updated{having a weight of $350$.}

\subsubsection{Approximation for the GPM problem} 

As shown in \S~\ref{subsec:complexity_measure}, in general the GPM problem is intractable, however, we show thereafter that if the weights of the hyperedges can be approximated by using the sum of the pairwise weights, then it is possible to obtain a (polynomial-time) solution to the GPM within a bounded-approximation of the optimal.

\begin{ourClaim} \label{prop:more_general_result}
Assume that for every $p \in P$ and $X \subseteq C$ with $1 \leq |X| \leq k-1$ so that $\forall c \in X, \{p,c\} \in E$, we have
$$\alpha_1(k) \cdot \sum_{c \in X} w(\{p,c\})  \leq w(\{p\} \cup X) \leq \alpha_2(k) \cdot \sum_{c \in X} w(\{p,c\}),$$
and let $A$ be an approximation algorithm for the one-to-many assignment problem with approximation ratio $r$, then $A$ provides an $r\cdot \alpha_2(k) / \alpha_1(k)$-approximation for the $(k,\Delta)-$GPM problem.
\end{ourClaim}

\begin{proof}
Let's suppose we use the construction explained in \S~\ref{subsec:assignment_pb} to build a one-to-many assignment problem, i.e., a bipartite graph $G_{\hbox{p}} = (P_k \cup C, E_\Delta)$ with $P_k = \{p^j \;|\; p \in P, 1 \leq j \leq k-1\}$ from a given $(k,\Delta)-$GPM instance over prosumer set $P$ and consumer set $C$.





Let $\mathcal{M}_{opt}$ be an optimal solution for the $(k,\Delta)$-GPM problem and $M_{opt}$ be an optimal solution for the one-to-one assignment problem in $G_{\hbox{p}}$. 
Let $\mathcal{M}_A$ be the hypergraph matching obtained by merging together all pairs sharing a copy of a vertex $p$ in the matching $M = A(G_{\hbox{p}})$ obtained by executing $A$ over the input $G_{\hbox{p}}$, the graph having been constructed in a way so that $\mathcal{M}_A$ is an answer to the $(k,\Delta)$-GPM problem. Reversely, let 
$M_{p}$ be the matching of the pairs obtained by breaking the hyperedges of $\mathcal{M}_{opt}$, thus forming edges from $G_{\hbox{p}}$.

Because algorithm $A$ has approximation ratio $r$, we have 
$$w(M) \geq \frac{w(M_{opt})}{r}.$$
Moreover, by the proposition's initial assumption, we have (using the lower-bound side) $$\alpha_1(k) \cdot w(M) \leq w(\mathcal{M}_A)$$ and (using the upper-bound side) $$w(\mathcal{M}_{opt}) \leq \alpha_2(k) \cdot w(M_{p}).$$ 
Since $w(\mathcal{M}_{opt}) \geq w(\mathcal{M}_A)$ and $w(M_{opt}) \geq w(M_{p})$ by their optimality property, we obtain
\begin{align*}
w(\mathcal{M}_A) &\geq \alpha_1(k) \cdot w(M) \\
&\geq \alpha_1(k) \cdot \frac{w(M_{opt})}{r} \\
&\geq \alpha_1(k) \cdot \frac{w(M_p)}{r} \\
&\geq \frac{\alpha_1(k)}{\alpha_2(k) \cdot r} \cdot w(\mathcal{M}_{opt}).
\end{align*}

Hence, $A$ provides an $r\cdot \alpha_2(k) / \alpha_1(k)$-approximation for the $(k,\Delta)$-GPM problem.
\end{proof}

\subsection{Pairwise Cost-saving Weights as Approximation for GPM Weights} \label{subsec:approx_weights}

Following the construction presented in \S~\ref{subsec:assignment_pb}, one can use algorithms to solve the one-to-one assignment problem to build a hypergraph matching and hence a partition of consumers and prosumers into independent energy communities. However, the partition will not lead in this case to the global maximum in terms of weight as shown in the example of Fig.~\ref{fig:matchings}(f). Nevertheless, one can bound how far from the optimal hypergraph matching the constructed matching is by building on the following proposition (cf. Proposition~\ref{prop:optimal_hypermatching}).   
%
%
\updated{In the following, the hyperedge weights correspond to the ones calculated based on cost savings following Definition~\ref{def:cost_savings}, \ie{} $w_t(G) = \var{gain}(G, [t-\tau, t])$ for a certain $\tau$ (e.g. one month); $t$ is omitted in the following. Neighborhoods are defined in Definition~\ref{def:neighborhoods}.}


\begin{ourClaim} \label{prop:ratio_bounds}
For every prosumer $p \in P$ and group $X \subseteq C$ made only of consumers, with $1 \leq |X| \leq k-1$, and so that all consumers of $X$ are located within distance $\Delta$ from $p$, \ie{} $\forall c \in X, c \in N_\Delta(p)$, we have
$$\frac{1}{|X|} \sum_{c \in X} w(\{p,c\})  \leq w(\{p\} \cup X) \leq \sum_{c \in X} w(\{p,c\})$$
where $w$ corresponds to cost-saving weights according to Definition~\ref{def:cost_savings}, \ie{} the financial gain of using P2P energy sharing for the community $G = \{p\} \cup X$.
\end{ourClaim}

\begin{proof} \updated{For ease of notation, let $w_X = \sum_{c \in X} w(\{p,c\})$.
\textbf{Lower bound}. In any community $G = \{p\} \cup X$, it is always possible to ignore some of the consumers upon optimizing for the community's cost (\ie{} adding more consumers to a community can only produce more benefits and not less).} Hence, we have $w(G) \geq w(\{p,c\})$ for any $c \in X$, that is $w(G) \geq \max_{c \in X} w(\{p,c\})$. Since $w_X \leq |X| \cdot \max_{c \in X} w(\{p,c\})$, we get $w(G) \geq w_X / |X| $.

\textbf{Upper bound}. The intuition behind the upper bound comes from the infeasibility to produce more cost saving as a group than the sum of each consumer individually working with the same prosumer $p$.
\updated{If it were possible, one would} be able to re-create an individual strategy for one of the consumers $c$ that would then beat the optimal individual strategy that produces a gain of $w(\{p,c\})$.
Essentially, each gain obtained by the community as a whole is either due to an efficient usage of $p$'s battery (irrelevantly of the consumers associated with $p$) or to a price difference between the price paid by using $p$'s local energy (either from the same hour or from $p$'s battery) and the current centralized grid's price.

\updated{Suppose one has an allocation $A$ of $p$'s resources over time in order to maximize the cost saving $w(G)$ for the community $G$, and let's note $w^A(G)$ the cost saving obtained following such allocation, \ie{} $w^A(G) \geq w^{A'}(G)$ for any other allocation $A'$ of $p$'s resources. $A$} provides, for each timestep, $p$'s decisions upon the quantity of energy to transfer into/from its battery system, and by consequence how much energy is bought/sold from/to the grid by the group while covering $X$'s aggregated consumptions. Hence, for any $c \in X$, one can deduce from $A$ an ``individual strategy'' $A_c$ following $A$ and counter-acting the energy consumption $\ell$ of $X \setminus \{c\}$ by selling $\ell$ kWh to the rest of the grid. 
Intuitively, the cost saving $w^{A_c}(\{p,c\})$ associated with the $A_c$ allocation cannot beat the best individual allocation for $c$ cooperating solely with $p$ as a community of size $2$, that is $w^{A_c}(\{p,c\}) \leq w(\{p,c\})$. \updated{We formally prove the last statement hereafter.}


Let's have a more focused look first at how the cost savings are calculated (\ie{} base cost minus community cost, over the time period $\mathcal{T}$):
\begin{equation}\label{eq:cost_G}
    w^A(G) = \sum_{x \in G} \var{bill}(x, \mathcal{T}) - \var{bill}_A(G, \mathcal{T}),
\end{equation}
with $\var{bill}_A(G, \mathcal{T})$ being the cost over $\mathcal{T}$ under $A$.
Following the strategy $A$, denote $\bm{bat}_A(p,t)$ the battery of $p$ at time $t$, $\bm{el_{in}}^A(t)$ the energy bought from the grid at time $t$ and $\bm{el_{out}}^A(t)$ the energy sold during the same hour.
Since we assume $A$ is a valid strategy, we have that $\var{el}_{cons}(h,t)$ for $h \in G$ is balanced by electricity coming from either $p$'s battery, the grid or $p$'s local production. Hence, noting $\bm{bat}_A(t) = \bm{bat}_A(p,t)-\bm{bat}_A(p,t-1)$ and $\bm{grid}_A(t) = \bm{el_{in}}^A(t) - \bm{el_{out}}^A(t)$, we have:
$$\sum_{h \in G} \var{el}_{cons}(h,t) = \bm{bat}_A(t) + \var{el}_{gen}(p,t) + \bm{grid}_A(t).$$


\paragraph*{Individual allocation of the community's benefits}

We can now make the consumer allocated energy strategies under $A$ explicit as follows. 
Order arbitrarily first the members of $G = \{p, c_1, ..., c_{k-1}\}$ and process the $c_i$'s in the same order (hereafter consider ``$p = c_0$''):
\begin{enumerate}
    \item Use in priority the energy from $\var{el}_{gen}(p,t)$;
    \item If $\var{el}_{cons}(c_i,t)$ is not yet covered by 1 and if there is still energy in the battery pool, i.e. $\bm{bat}_A(t) > 0$, then use it;
    \item If $\var{el}_{cons}(c_i,t)$ is not yet covered by 1 \& 2, then use energy from $\bm{grid}_A(t)$.
\end{enumerate}

With such explicit allocations, we can compute ``individual cost savings'', \updated{that is reusing the $w$-notation we set} $w^A(c_i) = \sum_{t \in \mathcal{T}} w^A_t(c_i)$ is the cost saving that we can attribute to $c_i$ and calculated as follows. We set $\bm{el_{out}}^A(c_i,t) = 0$ if $i > 0$, $\bm{el_{out}}^A(c_0,t) = \bm{el_{out}}^A(t)$ and $\bm{el_{in}}^A(c_i,t) = \bm{grid}_A(c_i,t)$ where $\bm{grid}_A(c_i,t)$ is the amount of energy drawn from $\bm{grid}_A(t)$ by the above allocation (potentially $0$ if $c_i$ only needed solar panel and/or battery energy during that timestep). Individual cost saving is then obtained as:
$$w^A_t(c_i) = \var{cost}(c_i,t) - \var{cost}_A(c_i,t)$$
where $\var{cost}(c_i,t)$ is the usual cost (no community, \updated{cf.~\S~\ref{subsec:LPsolver}}) and $\var{cost}_A(c_i,t)$ is $c_i$'s cost using $\bm{el_{in}}^A(c_i,t)$ and $\bm{el_{out}}^A(c_i,t)$ in the cost calculation \updated{(eq.~\ref{eq:cost})}.
%
Now, let's observe that we have from developping equation~\ref{eq:cost_G}:
\begin{align*}
    w^A(G) &= \left(\sum_{h \in G} \var{bill}(h, \mathcal{T})\right)- \var{bill}_A(G, \mathcal{T}) \\
    &= \sum_{0 \leq i \leq k-1} w^A(c_i).
\end{align*}
This holds because $$\var{bill}_A(G, \mathcal{T}) = \sum_{h \in G} \var{cost}_A(h) = \sum_{h \in G} \sum_{t \in \mathcal{T}} \var{cost}_A(h,t),$$ and irreverently on how the energy is locally allocated in each timestep, aggregated together, they will account for all $G$'s cost.

Now, let's show that $w^A(p)$ is negative, that is $p$'s contribution under our formulation is decreasing the total cost saving $w^A(G)$. 
In our allocation of resources, $p$ has priority on its battery and PV panels regardless if it is part of a community or alone.  
However, if $p$ were left alone, some $\var{el}_{gen}(p,t)$ or $\bm{bat}_A(t)$ could have been sold to the grid or stored instead of being used by the other peers forming the community $G$, which could eventually only decrease $p$'s cost (and cannot increase it by any mean).

Now, recall $w(\{p,c\})$ is the optimal cost saving for the pair $\{p,c\}$. Consider the strategy $A_c$ for the pair $\{p,c\}$ that follows $A$ except it sells to the grid the energy that was originally intended to cover the other consumers $h \in G \setminus \{p,c\}$. Since $A_c$ is not optimal in regards to the pair $\{p,c\}$, $w^{A_c}(\{p,c\}) \leq w(\{p,c\})$. Obviously thanks to the additional benefits coming from selling extra electricity, the cost under $A_c$ is strictly less than $\var{cost}_A(h)$ entailing $w^A(c) \leq w^{A_c} (\{p,c\})$. However, $A_c$ being a valid strategy for the pair $\{p,c\}$, one also gets that $w^{A_c}(\{p,c\}) \leq w(\{p,c\})$ which 
entails $w^A(c) \leq w(\{p,c\})$.


\updated{
}
\updated{We can now conclude our proof by putting together the last two claims:
\begin{align*}
    w^A(G) &= \sum_{0 \leq i \leq k-1} w^A(c_i) = w^A(p) + \sum_{1 \leq i \leq k-1} w^A(c_i) \\
    &\leq \sum_{1 \leq i \leq k-1} w(\{p, c_i\}). 
\end{align*}
\vspace{-0.1cm}
}
\end{proof}

In our evaluation, we compare the bounds shown in the Proposition~\ref{prop:ratio_bounds} with the ones computed in practical instances in Fig.~\ref{fig:ratios} and conclude that the theoretical bounds shown here are close to the ones observed in practice. \updated{At last, let us derive from the previous proposition an approximation bound for the GPM problem by using the optimal pairwise weights as approximation for the hyperedge weights.
Using in Proposition~\ref{prop:more_general_result} the optimal assignment over $G_{\hbox{p}} = (P_k \cup C, E_\Delta)$ that can be obtained in $\mathcal{O}(kn^3)$ with the Hungarian algorithm, and considering the bounds provided by Proposition~\ref{prop:ratio_bounds}, i.e., $\alpha_1(k) = 1/(k-1)$ and $\alpha_2(k) = 1$,
we obtain the following result.
}


\begin{corollary} \label{prop:optimal_hypermatching}
For cost-saving weights, it is possible to compute a $(k-1)$-approximate hypergraph matching solution to the $(k,\Delta)-$GPM in $\mathcal{O}(kn^3)$ time.
\end{corollary}

\section{Efficient Peer Matching Algorithms} \label{sec:algorithms}


\updated{We describe in this section algorithms that produce a solution (\ie{} a hypergraph matching) to the GPM problem with the input of the matching problem being a bipartite graph based on end-users location (as defined in~\S~\ref{subsec:pb_abstraction}).} 
We first present three algorithms: \textit{Round Robin}, \textit{Single Pass} and \textit{Classic Greedy} in \S~\ref{sec:algo_templates}, then explain in \S~\ref{subsec:weights} how to instantiate them using different weight functions.
\updated{We analyze here the time-complexity of the introduced matching procedures}; we note the weight function is constant in this analysis.
Dealing with the latter, for each algorithm and weight function, we provide an asymptotic analysis of the number of weight computations in order to produce the solution in \S~\ref{subsec:weights}.
\updated{Recall that following \S~\ref{subsec:peer_matching}, in order to} reduce the search space, we have restrained hyperedges to contain exactly one member of the set $P$ (a prosumer in our context) but the algorithms can be easily adapted if this constraint is lifted.

\subsection{Peer Matching Algorithms} \label{sec:algo_templates}

\subsubsection{Definitions} \label{defitinions}

Let $P \cup C$ be the input of a GPM problem of parameter $(k,\Delta)$. We assume that \textsf{weight}$(p_i, c_j, M_i, t)$ returns a weight associated to node $p_i \in P$ and node $c_j \in C$ and possibly using a partially computed set of nodes $M_i \subset C$ that are already associated with $p_i$, whereas $t$ indicates the current time-step (recall, weights are time-dependent). This function can be resolved by either:
\begin{enumerate}[leftmargin=*]
    \item executing a computationally-expensive procedure that relies on solving one or several LP-programs as defined in~\S~\ref{subsec:LPsolver} and~\ref{subsec:community};
    \item simply performing a look-up of a previously computed weight.
\end{enumerate}

When a weight is obtained in the first case above, we say it is \textit{computed}; such computation is the bottleneck of the matching procedures and is highlighted in the pseudocodes of the algorithms.
The input of the matching algorithms is the bipartite graph $H_\Delta = (P \cup C, E_\Delta)$ where $E_\Delta$ captures all neighborhoods at geographic distance $\Delta$ from each prosumer, \ie{} $E_\Delta = \{ (i,j) \in P \times C \;|\; \textsf{dist}(p_i, c_j) \leq \Delta \}$; an example of $E_\Delta$ restraining possible hyperedges is illustrated in Fig.~\ref{fig:matchings}(a)-(b). 
We highlight in the following remark that the presented  \textit{Peer Matching} algorithms are more general than the GPM problem. 

\begin{remark} \label{remark:gen_algo}
Note that Algorithms~\ref{alg:matching1},~\ref{alg:matching2} and~\ref{alg:matching3} do not rely on the geographical positions of the nodes but on the authorised matching pairs captured by $E_\Delta$, hence any bipartite graph can be given as input and the algorithms can be applied to a wider range of problems. 
\end{remark}

The function $\textsf{order}(P)$ sorts the set $P$ according to a predefined ordering, to be provided by the user (cf~\S~\ref{subsec:exp_setup} for the ones used in our evaluation).
As an example, the matchings obtained by the three algorithms are displayed in Fig.~\ref{fig:matchings}(d)-(f) when using the pairwise weights of Fig.~\ref{fig:matchings}(c) with processing order $p_1, p_2, p_3$.


\begin{algorithm}[t]
\footnotesize
\SetAlgoLined
\SetKwInOut{Input}{Input}
\SetKwInOut{Output}{Output}
\SetKwInOut{Upon}{Event}

\Input{
A bipartite graph $H_\Delta = (P \cup C, E_\Delta)$ and $k \geq 2$ 
}
\Output{$M$, a $k$-bounded bipartite hypergraph matching\;}
    \tcp{\color{comment-color}Initialization}
    \ForEach{$i \in P$}{ \label{alg1:line1}
        $M_i \leftarrow \emptyset$ \;
    }
    \ForEach{$j \in C$}{ 
        $S_j \leftarrow $ False \;
    }
    $\Psi \leftarrow \textsf{order}(P)$ \; \label{alg1:line8}
    
    \While{$\Psi \not= \emptyset$}{ \label{alg:while} 
        \ForEach{$i \in \Psi$}{ \label{alg:loopstart}
            $N \leftarrow \{ j \in C \;|\; \{i,j\} \in E_\Delta \land \lnot S_j \}$ \;
            \eIf{$N = \emptyset \lor |M_i| = k-1$}{
              $\Psi \leftarrow \Psi \setminus \{i\}$ \;\label{alg:prosumer_removed}
            }
            {
                \eIf{$|N| > 1$}{
                    \ForEach{$j \in N$}{
                        \colorbox{red!10}{$b_j \leftarrow$ \textsf{weight}$_t(p_i, c_j, M_i)$\;}
                    }
                    $\ell \leftarrow \argmax_{j \in N} \; b_{j}$ \;
                }
                {
                    $\ell \leftarrow N[1]$
                }
                
                $S_{\ell} \leftarrow $ True \;
                $M_i \leftarrow M_i \cup \{\ell\}$\;
            }
        }
    }    
\Return $\{M_i \;|\; i \in P\}$\;
\caption{Round Robin Matching Procedure}
\label{alg:matching1}
\end{algorithm}

\subsubsection{Round robin matching}


Algorithm~\ref{alg:matching1} builds a bipartite hypergraph matching by affecting to each prosumer one consumer at a time in a round-robin fashion. 
If no consumer can be affected to $p$ (either because $p$ has already $k-1$ consumers affected or no unmatched consumer is found within distance $\Delta$ of $p$), $p$ is skipped (and discarded from subsequent matching attempts).
Hence, all prosumers get at most 1 consumer each before a second iteration starts in affecting a second consumer to every prosumer.
When a prosumer has a choice of which consumer to pick, it greedily selects the consumer with highest weight; how the weights are settled is further described in \S~\ref{subsec:weights}.

\subsubsection{Single pass matching}


\begin{algorithm}[t]
\footnotesize
\SetAlgoLined
\SetKwInOut{Input}{Input/Output}
\Input{
as in Algorithm~\ref{alg:matching1}.
}
\tcp{\color{comment-color}Same as lines \ref{alg1:line1}-\ref{alg1:line8} in Algorithm~\ref{alg:matching1}}
    \ForEach{$i \in P$}{ 
        $N \leftarrow \{ j \in C \;|\; \{i,j\} \in E_\Delta \land \lnot S_j \}$ \;
        \eIf{$|N| \geq k$}{
            \ForEach{$j \in N$}{
            \colorbox{red!10}{$b_j \leftarrow$ \textsf{weight}$_t(p_i, c_j, M_i)$\; \label{alg2:weights}}
        }
        \While{$|M_i| < k-1$}{ \label{algo2:pick1}
            $\ell \leftarrow \argmax_{j \in N} \; b_{j}$ \;
            
            $M_i \leftarrow M_i \cup \{\ell\}$\; \label{algo2:pick2}
        }
        }{
            $M_i \leftarrow N$\;
            
        }
        
        \ForEach{$j \in M_i$}{
            $S_{j} \leftarrow $ True \;
        }
        
    }    
\Return $\{M_i \;|\; i \in P\}$\;
\caption{Single-Pass Matching Procedure}
\label{alg:matching2}
\end{algorithm}

Contrary to the previous one, Algorithm~\ref{alg:matching2} builds a matching in a single pass over the prosumers. For each prosumer $p$, the $k-1$ best available consumers (in terms of weight) are matched with $p$. 

\subsubsection{Classic greedy matching}

\begin{algorithm}[b]
\footnotesize
\SetAlgoLined
\SetKwInOut{Input}{Input/Output}
\Input{
as in Algorithm~\ref{alg:matching1}.
}
\tcp{\color{comment-color}Same as lines \ref{alg1:line1}-\ref{alg1:line8} in Algorithm~\ref{alg:matching1}}
    
    \ForEach{$\{i,j\} \in E_\Delta$}{
        \colorbox{red!10}{$b_{i,j} \leftarrow$ \textsf{weight}$_t(p_i, c_j, M_i)$\;}
    }
    \tcp{\color{comment-color}Sort all possible matching pairs by decreasing weights}
    $B \leftarrow \textsf{decreasing\_sort}( \{ b_{i,j} \;|\; \{i,j\} \in E_\Delta \} )$ \;
    \ForEach{$b_j \in B$}{
        \If{$|M_i| < k-1 \land \lnot S_j$}{
            $S_{j} \leftarrow $ True \;
            $M_i \leftarrow M_i \cup \{j\}$\;
        }
    }    
\Return $\{M_i \;|\; i \in P\}$\;
\caption{Classic Greedy Matching Procedure}
\label{alg:matching3}
\end{algorithm}

Algorithm~\ref{alg:matching3} is the ``classic'' greedy procedure for solving the assignment problem (cf.~\S~\ref{subsec:assignment_pb}) that, based on pre-computing all pairwise weights, sorts the prosumer-consumer pairs $(p,c) \in E_\Delta$ from highest to lowest, then associates consumers to prosumer whenever possible ($c$ not already matched and $p$ having less than $k-1$ affected consumers).
\updated{As it is well-known that the greedy algorithm provides a $2$-approximation to the assignment problem~\cite{duvignau2023greediness}, Propositions~\ref{prop:ratio_bounds} and~\ref{prop:more_general_result} entail an approximation bound of $2(k-1)$ on the quality of the output of the algorithm.}

\begin{corollary} \label{proposition_greedy} 
Algorithm~\ref{alg:matching3} provides a $2(k-1)$-approximation for the $(k,\Delta)-$GPM problem.
\end{corollary}

\subsubsection{Time complexity} \label{subsec:timecomplexity} 
We analyze here the time complexity of the three introduced matching algorithms, assuming a constant time for each weight calculation.  

\begin{ourClaim} \label{proposition123} 
Algorithm~\ref{alg:matching1},~\ref{alg:matching2} and~\ref{alg:matching3} run in respectively $\mathcal{O} \left( mk \right)$,  $\mathcal{O} \left( m \right)$ and $\mathcal{O} \left( m \log{m} \right)$ time.
\end{ourClaim}

\begin{proof} 
\textit{Algo.~\ref{alg:matching1}}: the main while-loop is executed at most $k$ times (after what all prosumers have been removed by line~\ref{alg:prosumer_removed}) and each inner for-loop goes through $|P|$ prosumers, each time calculating the best local choice using $\mathcal{O}(|N_\Delta(p)|)$ for prosumer $p$.
The inner for-loop thus takes total time $\mathcal{O}(|P| + |E_\Delta|) = \mathcal{O}(\avgN |P|)$.

\textit{Algo.~\ref{alg:matching2}}: line~\ref{alg2:weights} is executed $\mathcal{O}(|E_\Delta|)$ times whereas lines~\ref{algo2:pick1}-\ref{algo2:pick2} also take $\mathcal{O}(m)$ in total when using a selection and partition algorithm to find the $k-1$ highest unsorted weights in each neighborhood. 

\textit{Algo.~\ref{alg:matching3}}: this is the time needed to sort all weights, then going through the sorted list takes $\mathcal{O} \left( |P| \avgN \right)$ time.
\end{proof}

\subsection{Instantiation of the Weight Function} \label{subsec:weights}

We define here two ways to calculate the weights: either in a ``memoryful'' fashion, taking into account previous choices made by the matching algorithm, or in a ``memoryless'' one with constant weights only depending on the involved prosumer and consumer pair $(p_i,c_j)$ being examined.
For each variant, we propose two ways to calculate the weight for a given pair, either based on the energy cost produced by the pair working as a community in the recent past $[t-\tau,t]$ (running a single LP-solver as described in \S~\ref{subsec:LPsolver} and~\ref{subsec:community}) or based on the cost-saving produced by the pair over the same period $[t-\tau,t]$.
In the former case, the goal of the GPM problem being to minimize the total cost, we inverse the weights to keep a maximum-weight problem.
In all, four weight functions are defined as $\var{WX}_t(i,j)$ to instantiate the call to $\textsf{weight}_t(p_i, c_j, M_i)$.

\subsubsection{Memoryless weights}


Using memoryless weights (WA and WB) means that the weights are constant and independent of the matching procedure being run. This also means that building the matching is nothing more than solving the classic assignment problem (as described in~\S~\ref{subsec:assignment_pb}). An example of memoryless weights is given in Fig.~\ref{fig:matchings}(c) where a weight is given for each pair $(p,c) \in P \times C$. We set the memoryless weights as follows:
\begin{itemize}[topsep=0pt, leftmargin=*]
    \item[] $ \var{WA}_t(i,j) = -\var{bill}(\{p_i,c_j\},[t-\tau,t]); $
    \item[] $ \var{WB}_t(i,j) = \sum_{x \in \{p_i,c_j\}} \var{bill}(\{x\},[t-\tau,t])) + \var{WA}_t(i,j).$
\end{itemize}

\begin{ourClaim} \label{proposition4} 
Algorithms~\ref{alg:matching1},~\ref{alg:matching2} and~\ref{alg:matching3} need to compute $\mathcal{O}(m)$ memoryless weights to solve the $(k,\Delta)-$GPM problem.
\end{ourClaim}

\begin{proof} 
Note Algorithms~\ref{alg:matching2} and~\ref{alg:matching3} only call the weight routine once per prosumer-consumer pair, hence $\mathcal{O}(|E_\Delta|)$ weights are ever calculated.
By remembering previously computed weights, Algorithm~\ref{alg:matching1} needs also to compute each weight only once since the current matching is not involved when using memoryless weights. 
\end{proof}

Regarding the minimum number of weights to compute, the three algorithms differ. Algorithm~\ref{alg:matching3} needs to sort all the weights, so $|E_\Delta|$ weights are always computed regardless how the nodes are distributed in space.
However, for Algorithm~\ref{alg:matching1} and Algorithm~\ref{alg:matching2}, there exist cases where none of $\Omega(|P|^2)$ weights are computed\footnote{One can indeed construct the following example (for Algorithm~\ref{alg:matching1}): $p_i$ is positioned at $(0,i\cdot \varepsilon)$ and $c_j$ is positioned at $(0,\Delta+j\cdot \varepsilon')$  with $0 < \varepsilon < \frac{\Delta}{|P|}$ and $\varepsilon < \varepsilon' < \frac{\varepsilon (1+ |C|)}{|C|}$.}.
%
\updated{It is naturally possible to pre-compute all the $\mathcal{O}(n^2)$ memoryless weights based on some $E_\Delta$. However, we note} that the necessary number of weights to compute may be lower when using Algorithm~\ref{alg:matching1} and Algorithm~\ref{alg:matching2}.

\subsubsection{Memoryful weights}

We define two memoryful weight functions (WC and WD) to take into account the already matched pairs, hence trying to provide the most accurate answer to the bipartite hypergraph matching problem. The weights are set as follows:

\begin{itemize}[topsep=0pt, leftmargin=*]
    \item[] $\var{WC}_t(i,j) = -\var{bill}(M_i \cup \{p_i,c_j\},[t-\tau,t]);$
    \item[] $\var{WD}_t(i,j) = \sum_{x \in M_i \cup \{p_i,c_j\}} \var{bill}(\{x\},[t-\tau,t])) + \var{WC}_t(i,j).$
\end{itemize}

\begin{ourClaim} \label{proposition5} 
Algorithms~\ref{alg:matching2} and~\ref{alg:matching3} need to compute $\mathcal{O}(m)$ weights when using memoryful weights for solving the GPM problem of parameter $(k,\Delta)$; Algorithm~\ref{alg:matching1} needs to compute $\mathcal{O}(km)$ weights.
\end{ourClaim}

\begin{proof} 
Variants C and D now depend on previous matching.
As with memoryless weights, Algorithms~\ref{alg:matching2} and~\ref{alg:matching3} only call the weight routine at most once per edge in $E_\Delta$, hence $\mathcal{O}(|E_\Delta|)$ weights are ever used.
For Algorithm~\ref{alg:matching1}, the number of weight calculation is bounded by its time complexity, cf. Proposition~\ref{proposition123}.
\end{proof}

For Algorithm~\ref{alg:matching1}, pre-computation is not needed as each time the weight function is called, a weight calculation must occur (\ie{} $M_i$ is different for each iteration).
One can observe that all the algorithms compute significantly fewer weights than the number of hyperedges, that is bounded by $\mathcal{O} \left(|P| \cdot \binom{\delta}{k-1} \right) = \mathcal{O}(n \delta^{k-1}) = \mathcal{O}(n^k)$, with $\delta$ being the maximum size of the neighborhoods, \ie{} $\delta = \max_{p \in P} |N_\Delta(p)|$.
\looseness=-1

\subsubsection{Sampling for Reducing the Search Space}

One solution to further reduce the number of weights to compute is to sample the neighborhoods, before executing the matching procedure.
We propose here two sampling methods for a given sampling size~$s$. For neighborhoods containing less than $s$ consumers, \updated{we} keep them all, otherwise for neighborhood $N$: 
\begin{itemize}[leftmargin=*]
    \item \textbf{Random}: pick uniformly at random $s$ neighbors in $N$;
    \item \textbf{Greedy}: pick the $s$ consumers with largest consumptions in $N$ (expected to produce higher cost savings than the others).
\end{itemize}

Applying such a sampling step as a pre-processing of the neighborhoods reduces the number of weights to compute to only $\mathcal{O}(ns)$.

\begin{figure*}[t]
    \centering
    \vspace{-0.1cm}
    \includegraphics[width=0.9\linewidth]{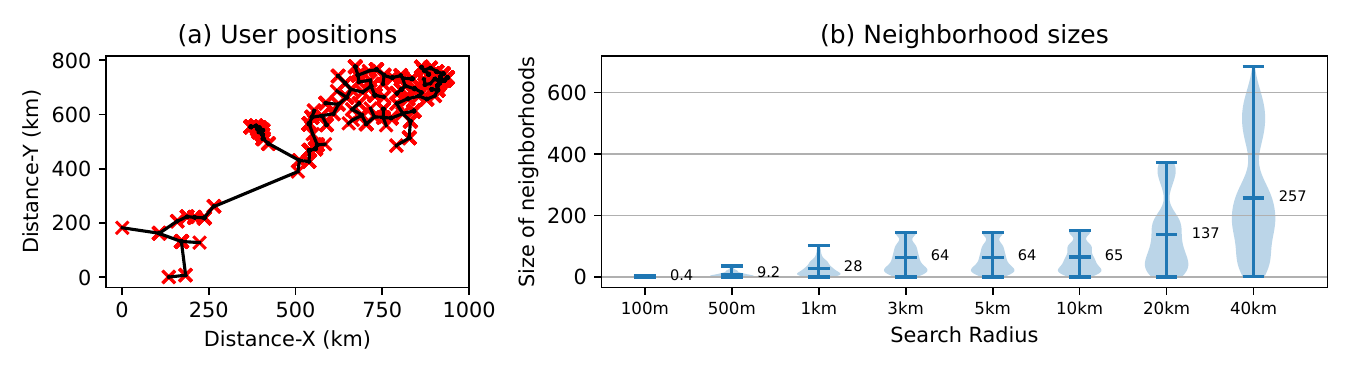}
    \caption{Dataset used in our evaluation: (a) Position of the users with an example of grid infrastructure relying them; (b) Size of the neighborhoods.}
    \label{fig:households}
\end{figure*}

\section{Experimental Study} \label{sec:evaluation}

\updated{We present in this section an experimental study comparing the performance of the proposed algorithms in terms of number of computed peer matching preferences (\ie{} hyperedge weights) versus the amount of cost saving produced by the computed matching (\ie{} the quality of the output matching).}

\subsection{Experimental Set-up} \label{subsec:exp_setup}

\subsubsection{Energy Data}
Consumption profiles of $2221$ real households are used in this study, originating from \cite{nyholm2016solar}. Each trace contains electricity consumption measures on a hourly basis for a different household for one year. 
Households are equipped with energy resources as follows: 10\% of them are prosumers with only PV panels, 10\% are prosumers with both PV panels and a battery, the rest are usual consumers. Capacities for energy resources are set for each prosumer relative to its average electricity consumption following current installed PV capacities (based on the Open PV 19 dataset~\cite{nrel2019open}).
An hourly solar profile is used here assuming same roof panel orientations~\cite{norwood2014geospatial}.
We use in our study the original city position of the households (obtained through using their respective zipcode), where we added a small random offset of up to 1km, see Fig.~\ref{fig:households}(a).
The hourly electricity price profile is based on the output from a European scale dispatch model~\cite{goransson2014linkages}.
Taxes are set to $\var{tax} = 25\%, \var{el}_{tax} = 6.9$~€ cents/kWh and $\var{el}_{net} = 0.58$~€ cents/kWh. 

\label{subsec:algo_searchradius}

\subsubsection{Algorithms} \label{exp:algo}
We compare in this study the three procedures defined in~\S~\ref{sec:algo_templates}: \textit{Round Robin}, \textit{Single Pass} and \textit{Classic Greedy} algorithms, \updated{as well as the hypergraph matching computed based on the optimal pairwise matching as described in \S~\ref{subsec:assignment_pb}}. 
For the first two, the order in which prosumers are processed can take three forms: \textit{Increasing} or \textit{Decreasing} order of their average energy consumptions, or decreasing order of their energy resources (\textit{Resource}); the latter is obtained by adding together both PV capacity (kWp) and battery capacity (kWh). 
Single Pass and Classic Greedy algorithms can be tuned to use two variants of the weight function (as explained in~\S~\ref{subsec:weights}): cost-based weights (WA) or saving-based weights (WB).
In addition to the above two variants (referred as ``memoryless weights''), the Round Robin algorithm is tested with the two memoryful counterparts: cost-based memoryful weights (WC) or saving-based memoryful weights (WD).
We \updated{hence} tested the 16 different combinations of algorithms, prosumer orders and weight functions of Table~\ref{tab:procedures} (cf. first three columns of Table~\ref{tab:algos}).
The experiments were run on a high-end server (Intel Xeon E5 2650 CPU, 64GB RAM) where computing a single weight takes about 2sec.

\begin{table}[b]
    {
    \caption{Tested parametrized procedures.}\label{tab:procedures}
    \centering\small
    \begin{tabular}{c|c|c}
        \textbf{Algorithm} & \textbf{Processing Order} & \textbf{Weight Function}  \\
        \hline
        \hline
        \multirow{3}{*}{Round Robin} & Increasing & WA, WB, WC, WD \\
        & Decreasing & WA, WB, WC, WD \\
        & Resource & WA, WB \\
        \hline
        \multirow{2}{*}{Single Pass} & Decreasing & WA, WB \\
        & Resource & WA, WB \\
        \hline
        Classic Greedy & -- & WB \\
        \hline
        Optimal Pairwise & -- & WB \\
    \end{tabular}
    }
\end{table}


\subsubsection{Search radius}
We set for the experiments 8 different search radii: 100m, 500m, 1km, 3km, 5km, 10km, 20km and 40km. 
For each search radius, Fig.~\ref{fig:households}(b) presents the distribution of the size of the nodes' neighborhoods (only based on the node positions). 
Since the search radii of 5km and 10km produce very similar neighborhood sizes as 3km, we omit them in the rest of the evaluation.

\subsubsection{Pairwise versus hyperedge weights}

\begin{figure}[t]
    \centering
    \includegraphics[width=\linewidth]{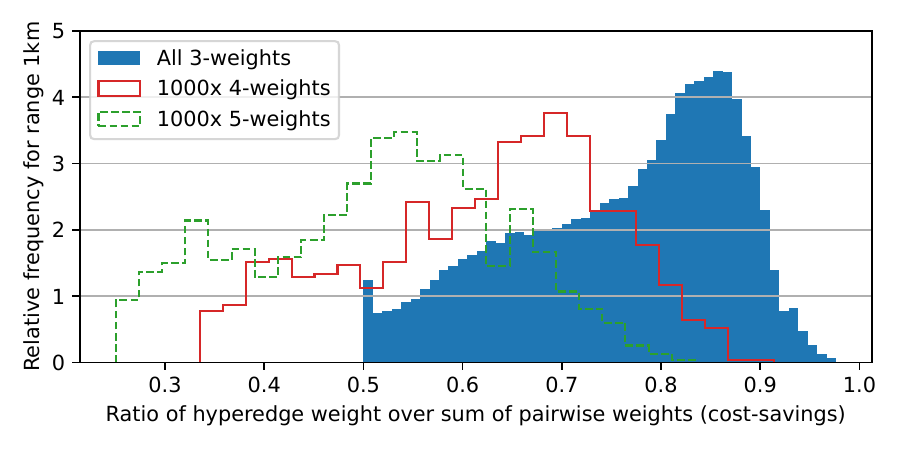}
    \caption{\updated{Comparison of hyperedge weights with the sum of pairwise weights for memoryless cost-saving weights.}}
    \label{fig:ratios}
\end{figure}

\updated{
All studied memoryless algorithms manipulate only the pairwise weights to infer their matching decisions.
Following Proposition~\ref{prop:ratio_bounds}, we know that when using cost-saving weights, we can upper- and lower-bound all hyperedge weights in relation to the sum of the pairwise weights. 
Fig.~\ref{fig:ratios} explores what is the distribution of the ratio between the hyperedge weights (the ``real weights'') and the sum of pairwise weights (the ``approximate weights'' for us as presented in \S~\ref{subsec:approx_weights}), i.e., the parameter $w(\{p\} \cup X) / \sum_{c\in X} w(\{p,c\})$. The figure illustrates how close to the theoretical bounds of $\mathbf{1 / (k-1)}$ (lower bound) and \textbf{1} (upper bound) the actual values are, when considering a range of 1km (using about 80k weights for possible triplets, 10k pairwise weights and 1000 random samples for 4- and 5-weights).}
\updated{
One may observe that (i) the bounds obtained in Proposition~\ref{prop:ratio_bounds} are very close to being tight for $k=3$, (ii) the lower-bound is quite close to the lowest found ratios and that, (iii) the ratios are in general most often between 0.5 and 0.9.
This indicates that pairwise weights provide indeed a suitable approximation for the weights of larger groups.
}

\subsection{Comparison of Peer Matching Algorithms}

We shall hereafter compare the performance of the introduced algorithms based on our two performance measures: \textit{quality of the matching} (in terms of cost saving achieved by the communities) and \textit{number of weight calculations} (main computational bottleneck) for the 16 parametrized procedures of Table~\ref{tab:procedures}.


\subsubsection{Quality of the matching}

\begin{table*}[t]
    \centering
    \setlength\tabcolsep{3.5pt}   
    \renewcommand{\arraystretch}{1.4}
    \caption{Comparison of the different algorithms, processing orders and weight functions~$\star$.}
    \label{tab:algos}
    \setlength\tabcolsep{4pt}   
\renewcommand{\arraystretch}{2}
\includegraphics[width=\textwidth]{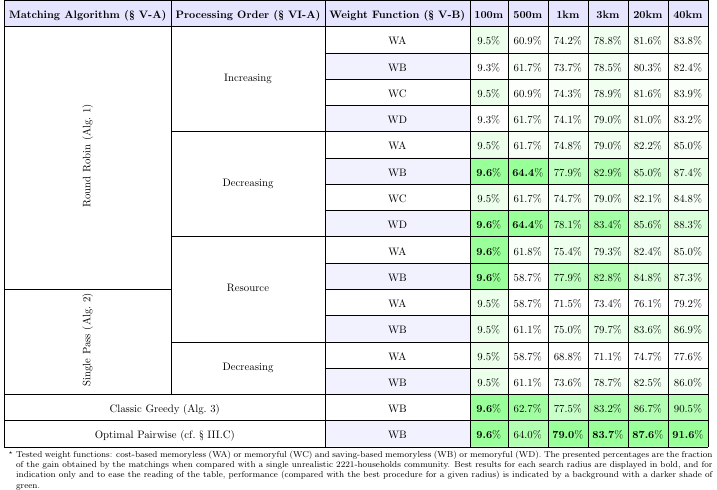}
\end{table*}

The solution obtained using the studied algorithms is summarized in relative terms in Table~\ref{tab:algos}.
The results are calculated upon matching together 445 prosumers (half of them having also a battery system) with 1776 consumers into groups of size at most $k = 5$ containing exactly one prosumer each and running the LP-solver (cf.~\S~\ref{subsec:LPsolver}) over a year of data to obtain each community's yearly cost (used as the hyperedge weight).
Recall the hyperedge weight is the basis of the \texttt{weight} function in Algorithms~\ref{alg:matching1},~\ref{alg:matching2} and~\ref{alg:matching3}, following the four different ways to calculate weights as presented in~\S~\ref{subsec:weights}. 

In our scenario, up to 150k~€ can be saved when all the nodes cooperate in a single community (equivalent of setting $\Delta = k = \infty$), from which up to 91.6\% can be recovered by using $\Delta = 40$km and $k = 5$.
In terms of quality of the solution, the Single Pass algorithm using decreasing order and cost-based weights is clearly outperformed by the rest, whereas the three best performing ones are the Optimal Pairwise, the Classic Greedy and the Round Robin algorithm using decreasing order and WD (all with cost-saving weights).
Closely behind, we note the high performance of both the Round Robin (decreasing order, WB) and the Single Pass (WB) algorithms, both computing significantly less weights than the algorithms that require the full pairwise weight list (Optimal Pairwise and Classic Greedy).

\subsubsection{Number of weight calculations}

\begin{figure}[t]
    \centering
    \vspace{-0.25cm}
    \includegraphics[width=\linewidth]{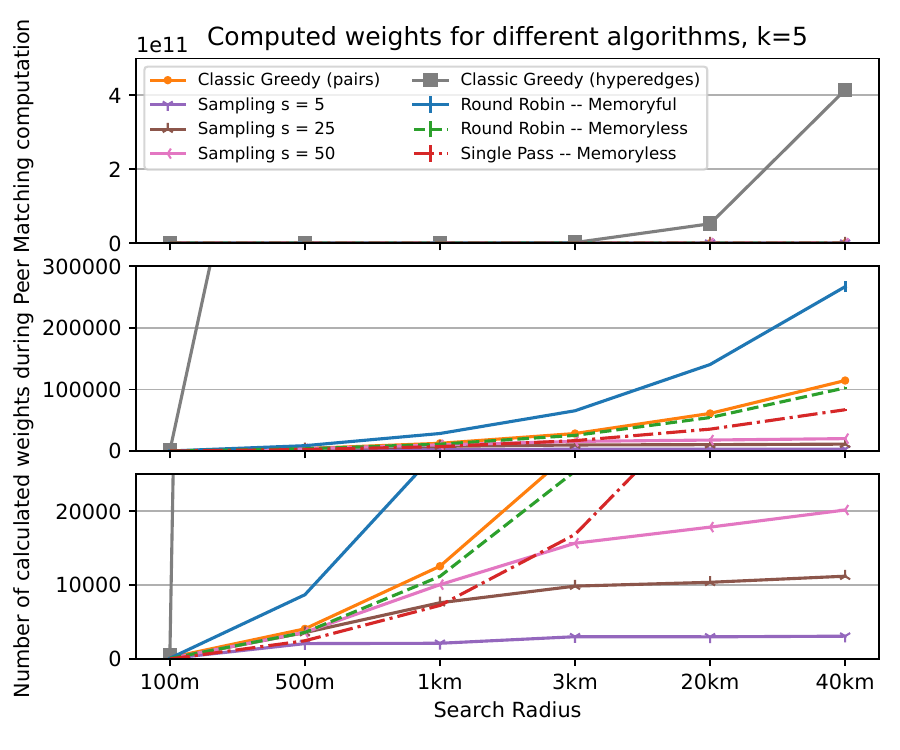}
    \caption{Average number of calculated weights for the different matching algorithms.}
    \label{fig:weights}
\end{figure}

Fig.~\ref{fig:weights} presents the number of weights\footnote{Variability is negligible in our evaluation: the average value is displayed with error bars, if any, indicating min-max values accounting for the different tested weight functions and prosumer orders. \updated{Both Optimal Pairwise and Classic Greedy require the same number of weights, $|E_\Delta|$.}} calculated along the computation of the matchings of Table~\ref{tab:algos}. 
Recall a weight calculation corresponds to running an expensive LP-solver (\S~\ref{subsec:LPsolver}) and may correspond to calculate the cost or cost-saving of a community of size $2$ (for all matching algorithms except Round Robin) up to $k$ (for memoryfull weights of the Round Robin matching).

All algorithms (except the sampling ones) display a quadratic increase in number of weights as the search radius increases, with varying slopes.
Confirming the theoretical upper bounds given in Section~\ref{sec:algorithms}, memoryful weights entail a large computational overhead.
On the contrary, memoryless weights reduce the burden on calculating weights with the Single Pass algorithm only calculating about 57\% of all prosumer-consumer pairs, and less than 25\% of the memoryful weights.
As expected, sampling $s$ values among the neighborhoods further reduces the weight computations to a bit less than $ns$ in total.
\updated{One can observe that all tested algorithms} compute several orders of magnitude less weights than the number of possible communities, at about 414 billion possible communities for $\Delta =$ 40km, that would be required to compute in order to run a greedy algorithm on the possible hyperedges (instead of on the prosumer-consumer pairs).

\subsubsection{Neighborhood size}


\begin{figure}[t]
    \centering
    \includegraphics[width=\linewidth]{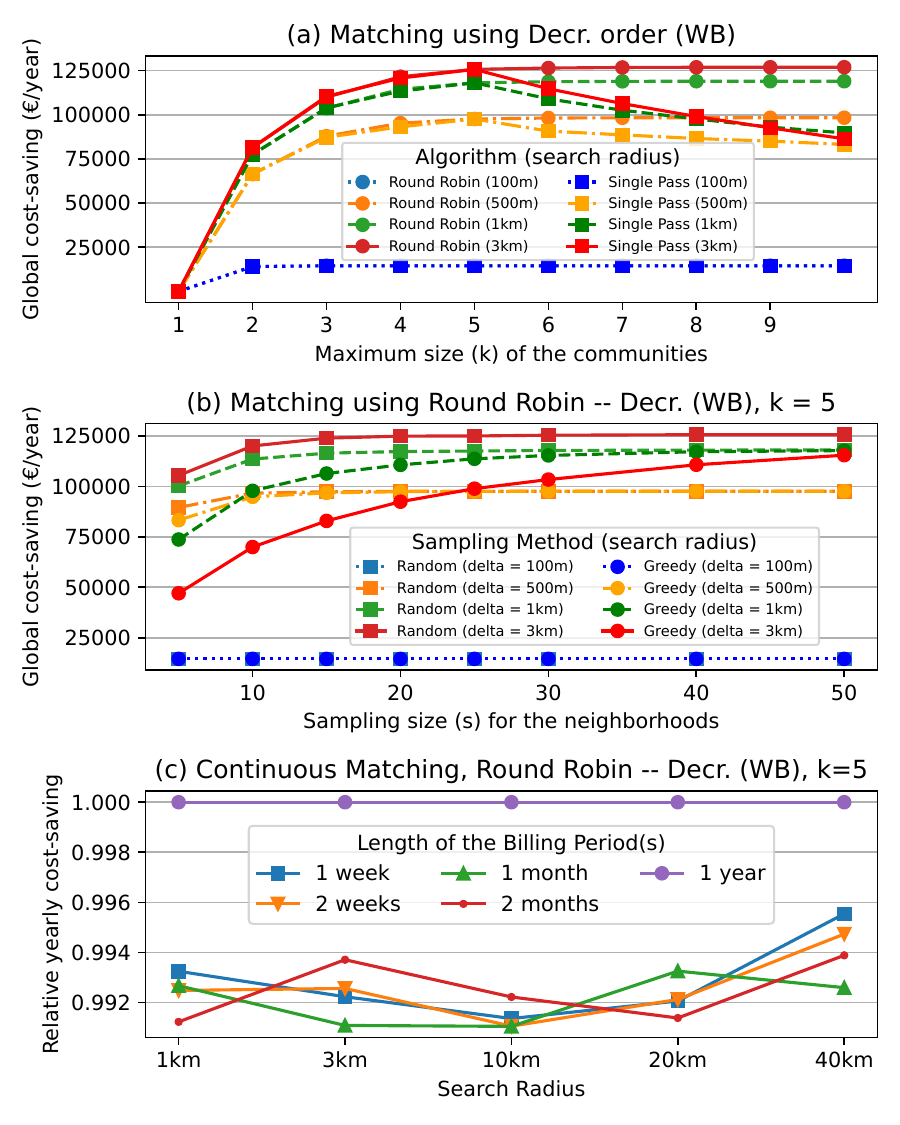}
    \caption{Sensitivity analysis showing the impact of (a) different maximum neighborhood sizes, (b) pre-sampling the neighborhoods and (c) shortening the billing period.}
    \label{fig:sensitivity}
\end{figure}

The neighborhood size providing the best trade-off depends on the search radius, as illustrated by Fig.~\ref{fig:sensitivity}(a) showing the cost saving produced using different bounds on the community size ($k$). 
In our evaluation scenario, using $k=5$ provides a cost only 0.01\% inferior than using $k = 10$ for the Round Robin algorithm. 
For the Single Pass algorithm, using neighborhood sizes above $5$ is even detrimental as the first prosumers get matched with a large amount of peers which removes the possibility for later processed prosumers to get consumers at all.
We note choosing $k$ should follow the distribution of prosumers in the pool. 

\subsubsection{Weight sampling}


We analyze the effect of weight sampling in Fig.~\ref{fig:sensitivity}(b).
The two strategies to sample a fixed number of potential partners in each neighborhood are compared: either perform a random or a greedy selection.
The result clearly favors a random selection which is very efficient even for small sampling size (\updated{here, $2k$ = 10 to $4k$ = 20 neighbors}), whereas the greedy solution only catches up for large sampling size (both converges to the non-sampling algorithm when $s$ reaches $\delta$, the maximum size of the neighborhoods).

\subsubsection{Continuous matching} \label{subsec:continuous}



Impact of performing the peer matching on smaller intervals than one year is studied in Fig.~\ref{fig:sensitivity}(c). 
The figure displays in relative term the cost saving achieved when performing the peer matching every week, second week, month, and second month, with all compared to the one year matching used so far (set to 1).
The cost savings obtained with all smaller periods decrease compare to using a year of data, but only to a very small extend, all reaching at least 99\% of the yearly saving.
Even though shorter interval communities could potentially gain from seasonable changes in finding ideal trading partners among fellow peers, this experiment confirms that long-term communities are very robust in terms of cost-saving.

\subsection{Summary of the Results}

In short, our experimental evaluation advocates that: 
\begin{enumerate}[leftmargin=*]
    \item to maximize cost-efficiency, use Round Robin (decreasing order, WD) for a small search radius and Classic Greedy or Optimal Pairwise for larger search spaces; 
    \item memoryful weights do not provide a significant advantage in terms of cost-saving that justifies the induced overhead, with Round Robin (Decreasing or Resource order, WB) algorithm providing the best trade-off overall in our experimental scenario;
    \item applying a random pre-sampling of each neighborhood of \eg{} $s=20$ weights for $k=5$ reduces drastically the computational overhead with little impact on the cost-efficiency, with further reductions obtained using the Single Pass algorithm. 
    \updated{For a large search radius, the latter only requires a very small fraction of the weight computations of a brute-force approach enumerating all possible communities (about e.g. $2 \cdot 10^{-8}$ for $\Delta$=40km).}
\end{enumerate}

\section{Related Work} \label{sec:related_work}

P2P energy sharing has been in the focus of numerous research works in recent years (cf. the comprehensive survey \cite{tushar2021peer} and references therein). Among those works, one can distinguish two types of P2P sharing communities.
Short-term coalitions (lasting for \eg{} 10min) are formed to cover a single timeslot where game-theoretic approaches are used to optimize the individual gains among other things~\cite{tushar2019grid}.
On the contrary, long-term coalitions~\cite{long2018peer, heinisch2019organizing, chau2023approximately, duvignau2021benefits} seek to form communities that will take coordinated decisions within the same group of peers over months to years. 
Often also associated with geographical closeness, the long-term communities present many advantages from an infrastructure point of view, being able to better regulate local load-balancing of distributed energy resources, providing higher local self-sufficiency and/or reducing the local peak demand.
Research on optimizing the gain outcome from the peering process using constrained optimization has focused on different aspects: how communications are handled~\cite{long2018peer}, reaching stable partitions~\cite{chau2019peer,chau2023approximately}, finding optimal resources based on community sizes~\cite{heinisch2019organizing}, privacy aspects and amount of data transmitted over the network~\cite{khodabakhsh2019prosumer,duvignau2020small,duvignau2021benefits}, etc. 
Small neighborhoods have been shown to provide a high share of possible gain while being favorable in terms of data exchanges~\cite{duvignau2021benefits} and different matching algorithms appear in~\cite{chau2019peer,chau2023approximately}, based upon stable partitions (\ie{} nodes with self-interest) and on a cost-sharing mechanism known beforehand.
However, the algorithmic problem behind the formation of localized and long-term P2P energy sharing communities with a global objective has not been studied before to the best of our knowledge, neither analytically nor from a practical point of view.

\updated{
The solutions explored in this paper relate to the hypergraph matching problem, especially to the works proposing mechanisms to compute practical solutions.
The hypergraph matching problem does not permit a polynomial time $o(k/\log k)$-approximation unless P = NP~\cite{HazanSS06}.
All the best heuristic algorithms that have been proposed to solve the problem are based on the notion of ``local search''~\cite{arkin1998local}.
Local search consists in incrementally improving a starting solution by performing a series of small changes (typically switching membership of a node from one partition to another), that is known to be competitive to the optimal algorithm.  
A local search algorithm is used for example in~\cite{chandra2001greedy} to solve greedily the weighted $k$-set problem while showing that the introduced algorithm is $2(k+1)/3$-competitive.
One of the most competitive algorithms for the hypergraph matching problem is given by Berman~\cite{berman2000d}, providing a $(k+1)/2$-approximation algorithm. 
The algorithm is based on finding independent set in $d$-claw free graphs, and it requires a series of prepossessing steps in order to solve the hypergraph matching problem, \eg{} \cite{zhang2019virtual} makes these steps explicit to obtain a hypergraph matching algorithm from Berman's algorithm. 
Other works~\cite{wang2016hypergraph,zhao2020hypergraph} have used a similar pipeline to solve greedily the hypergraph matching problem based on local search.
Recent improvements on Berman's algorithm are given in~\cite{Neuwohner21,Neuwohner22,Neuwohner23}.
MWHM has also been studied in a distributed setting~\cite{cui2007distributed}, as well as under ``$b$-matching'' generalizations~\cite{parekh2014generalized} where nodes can appear in $b$ different hyperedges instead of one only. 
}

All local search-based algorithms, however, assume access to the full list of weights since their starting point is the output of what we referred here as the ``Classic Greedy'' algorithm. Also, those algorithms have a time complexity that is exponential in $k$~\cite{chandra2001greedy}.
In~\cite{ma2016hypergraph}, computing all the weights was already identified as a computational bottleneck when the number of possible hyperedges is large, proportional to $\mathcal{O}(n^k)$ when there are $n$ nodes in the input hypergraph. 
The authors hence propose a heuristic algorithm specific to their problem that only necessitates to compute $\mathcal{O}(kn^2)$ weights using an algorithm of time-complexity $\mathcal{O}(kn^3)$.
For the number of computed weights, this matches the same bound shown here when using memoryful weights and we further reduce it to only $\mathcal{O}(n^2)$ weights while introducing the notion of ``memoryless weights'' \updated{and exploiting pairwise weights}.
All the algorithms studied in this work have also noticeable smaller time complexities, namely $\mathcal{O}(k n^2)$, $\mathcal{O}(n^2)$ and $\mathcal{O}(n^2\log n)$.
Note the problem studied here adds three additional constraints to hypergraph matching: requirement for a bipartite solution, spacial constraints and a challenging environment where the weight of a hyperedge is computationally expensive to obtain.

\updated{
Recently, Duvignau and Klasing~\cite{duvignau2023greediness} demonstrated several analytical results for achieving bounded approximations for the assignment problem (\ie{} bipartite graph matching) using greedy algorithms that focus on both producing a good matching while inspecting only $\mathcal{O}(n)$ weights.  
Their analysis assumes that the input graph is processed in an heuristic order where edges associated with higher weights are usually processed earlier by the matching algorithm.
This completely matches our setting where prosumers with larger renewable resources are more likely to produce more cost-saving benefit, and thus their results entail that formal performance guarantees can also be derived for the heuristic algorithms that we have introduced in \S~\ref{sec:algorithms}.
}



\section{Conclusion} \label{sec:conclusion}

This paper studies the peer matching problem to participate in P2P energy sharing.
We introduce the Geographical Peer Matching problem within a well-defined mathematical framework setting the problem as a hypergraph matching problem with a bounded search radius $\Delta$ and output partitions of size up to $k$.
This allows known approximation algorithms for the weighted $k$-set problem to be ported to our energy sharing setting.
To provide an efficient solution to the problem, we introduce and analyze three different matching algorithms that do not require to compute all $\mathcal{O}(n^k)$ possible weights (each requiring the run of a \updated{computationally expensive} LP-solver) but only at most $\mathcal{O}(kn^2)$. 
\updated{Pre-sampling the weights can further reduce the number of weight calculations to $\mathcal{O}(sn)$.}
The introduced algorithms are both scalable in terms of required computation and efficient in terms of the quality of the computed solution. 
Our extensive experimental study indeed shows that up to $91.6$\% of the benefit of a very large community (\ie{} of size $k = n$ and unbounded geographical diameter $\Delta = \infty$) can be obtained by limiting communities to $5$ nodes only using a small bounded geographical search radius. 
We also provide optimizations that, even though they do not change asymptotic behaviors of the proposed algorithms, are shown in this work to yield a practical computational advantage without sacrificing the quality of the produced solution.
We expect that as the introduced algorithms are more general than our specific problem, they can also be useful in other contexts where it is challenging to compute a solution to the hypergraph matching problem.

Some practical aspects for further research are (1) how to push parts of the matching computation towards the end-users for a more edge-friendly solution, for saving data transfers and caring about the privacy perspective~\cite{duvignau2019streaming,havers2020driven}, while transitioning from batch-based to the online analysis required by today's smart metering infrastructure~\cite{van2018echidna,van2018locovolt}, and (2)~how to update the matching dynamically and maintain a stable network through the arrivals and departures of peers~\cite{duchon2014local,georgiadis2013overlays}, possibly building on and adapting previous distributed and adaptive algorithms for matching with preferences~\cite{georgiadis2012adaptive,khan2016efficient}.


\bibliographystyle{elsarticle-num}
\bibliography{bib} 

\end{document}